\documentclass[aps,
10pt,
pra,
groupedaddress,
superscriptaddress,
twocolumn,
showkeys
]{revtex4}

\newcommand{\RNum}[1]{\uppercase\expandafter{\romannumeral #1\relax}}

\usepackage{hyperref}
\hypersetup{
    colorlinks=true,
    linkcolor=blue,
    urlcolor=cyan,
    citecolor=magenta
    }

\usepackage{subfigure}
\usepackage{amsfonts,amssymb,amsmath}
\usepackage[]{graphics,graphicx}
\usepackage{amsthm}
\usepackage{tikz}
\usepackage{algorithm}
\usepackage{algpseudocode}
\usepackage{dsfont}
\usepackage{dcolumn}

\newtheorem{theorem}{Theorem}[section]

\newtheorem{corollary}{Corollary}[section]

\def\braket#1#2{\langle #1|#2 \rangle}
\def\bra#1{\langle #1|}
\def\ket#1{| #1\rangle}

\def\bs#1{\boldsymbol{#1}}

\def\Tr{\mathop{\rm Tr}}

\def\eref#1{(\ref{#1})}

\newcommand{\mc}[1]{\mathcal{#1}}
\newcommand{\mb}[1]{\mathbb{#1}}
\newcommand{\mr}[1]{\mathrm{#1}}

\newcommand{\te}[1]{\text{#1}}
\newcommand{\md}[1]{\mathds{#1}}

\begin{document}

\title{A new general quantum state verification protocol  by the classical shadow method}

\author{Xiaodi Li}
\email{lixiaodi@fudan.edu.cn}

\affiliation{State Key Laboratory of Surface Physics and Department of Physics, Fudan University, Shanghai 200433, China}
\affiliation{Institute for Nanoelectronic Devices and Quantum Computing, Fudan University, Shanghai 200433, China}	
\affiliation{Center for Field Theory and Particle Physics, Fudan University, Shanghai 200433, China}

\begin{abstract}
Verifying whether a quantum device produces a specific quantum state is a fundamental task in many applications of modern quantum technologies. In the conventional framework of quantum state verification, designing an optimal or efficient protocol for each type of state typically requires intricate, state-specific customization. Recently, Hsin-Yuan Huang et al. introduced a novel approach known as the shadow overlap protocol \cite{10756060_huang}, which leverages classical shadows to efficiently verify multiple classes of quantum states simultaneously.
In this work, we propose a new verification protocol that integrates key ideas from both the conventional framework and the shadow overlap protocol. To this end, we first reformulate the shadow overlap protocol using the formalism of hypothesis testing, which also underpins the conventional approach, and analyze the similarities and differences between the two. Our framework extends the capabilities of the shadow overlap protocol while addressing some of its limitations, yielding improved sample complexity and a more natural treatment of structured quantum states. We demonstrate the effectiveness of our protocol through applications to GHZ states and stabilizer states.
\end{abstract}

%
\keywords{Quantum state verification, classical shadow, Pauli measurements, Stabilizer states, GHZ states}

\maketitle



\section{Introduction}

With the rapid developments of quantum technologies, people can manipulate much more complex and larger quantum systems, whose sizes are out of reach of the simulations of classical computers. Hence, ensuring a quantum state prepared by a quantum device is exactly an interesting target state has become an urgent problem. Such problem, called \emph{quantum state verification} (QSV) or certification, has wide implications in quantum computation, quantum communication, and other quantum technologies.  

The traditional quantum state tomography methods can't be used to solve the QSV problem, as they need exponential copies of the prepared states with respect to the number of qubits to obtain a classical description of the state. To circumvent such a predicament, Sam Pallister et al. proposed a general framework, called the PLM framework, for verifying the prepared quantum states in \cite{Pallister_2018}. An explicit verification strategy in this framework consists of a sequence of tests constructed by local projective measurements with binary outcomes, if the prepared states pass all such tests, we will accept the statement that the device produces the target state correctly. The sample complexity, i.e., the number of prepared states, of such a general framework is approximately proportional to the inverse of the spectral gap of an operator corresponding to the strategy, called the \emph{strategy operator}. Hence, applying this framework to verify a specific type of state reduces to constructing an explicit protocol or a strategy operator, which has the maximal spectral gap.

Prior to the introduction of the PLM framework, several related works \cite{Hayashi_2006, Hayashi_2009} had already explored the verification of maximally entangled states. Notably, the authors of \cite{Hayashi_2009} first addressed the quantum state verification under the assumption that the states produced by the quantum device are independently and identically distributed (i.i.d.). In contrast, the PLM framework considers a more general setting in which the states may be adversarial and are therefore not necessarily i.i.d.
After the PLM framework, there are also many developments in the direction of QSV. The optimal protocol (strategy operator) with the maximal spectral gap has been constructed for states, like the Bell states, the maximally entangled states, the bipartite pure state, the GHZ states and the stabilizer states \cite{Pallister_2018,Zhu_2019,Yu_2019,Li_2020, Dangniam_2020}, and some near-optimal or efficient protocols are also proposed for some more complicated states like the W states, the Dick states, the phased Dick states, graph states and hypergraph states \cite{Liu_2019, Li_2021,PhysRevApplied.12.054047, Chen_2023, Zhu_2024}. 
In addition to the mentioned protocols, there are also some works trying to extend the PLM framework to more complicated or more practical circumstances, like the LOCC measurement \cite{Yu_2019, Wang_2019, Liu_2019, Yu_2022, Li_2021}, the adversarial scenario \cite{PhysRevA.100.062335, PhysRevLett.123.260504}, the blind measurement-based quantum computation \cite{Li_2023}. 

However, most previous works have focused on constructing verification protocols tailored to specific types of quantum states, limiting their applicability. A protocol designed for one class of states typically cannot be extended to others. Given the vast number of state classes with special structures, it is impractical to design optimal or efficient protocols for each type individually. Moreover, states with special structures constitute only a small subset of the full Hilbert space, leaving the problem of verifying generic states—those lacking special structures—as an open and challenging question.

Fortunately, Hsin-Yuan Huang et al. proposed a new, different protocol in \cite{10756060_huang}, which can verify almost all Haar random pure quantum states and some special states efficiently. 
This protocol, called the \emph{shadow overlap protocol} (SOP), unifies the ideas of randomized Pauli measurement and the classical shadow tomography \cite{Elben_2022, huang2020predicting}, then possesses the advantages of application in practical experiments.    
However, the SOP also has certain limitations. First, for some structured states such as GHZ states, applying the SOP requires additional unitary transformations, which are not inherently integrated into the protocol and must be handled separately. Second, the protocol relies on a complicated operator constructed by querying the amplitudes of the target state in the computational basis, which lacks a clear physical interpretation. 

In this paper, we explore the similarities and differences between the SOP and the PLM framework, and propose a new protocol, so-called the directly partial shadow overlap (DPSO) protocol, which remedies some shortcomings of the SOP by combining the ideas of both frameworks.
More explicitly, we reinterpret the SOP in the view of hypothesis test and compare it with the PLM framework. 
Then we construct the DPSO protocol within the same framework as SOP, under the assumption that the states generated by the quantum device are i.i.d.
Compared with the SOP, our protocol has lower sample complexity and broader applicability, being able to handle GHZ states or other special states more naturally. When the level of our protocol is equal to $1$, the level-$1$ DPSO protocol includes the SOP, so our protocol also retains the capability to verify almost all Haar random pure states. 


The organization of our paper is following. In Section \ref{sec_pre}, we introduce the PLM framework of QSV briefly. In Section \ref{sec_sop}, we reformulate the shadow overlap protocol in the formalism of hypothesis testing, then discuss similarities and differences between two frameworks in details. Next, in Section \ref{sec_protocol}, we propose our protocol explicitly and compute the sample complexity. Finally, in Section \ref{sec_stabilizer}, we apply our protocol to the stabilizer state and the GHZ state.


\section{Preliminaries} 
\label{sec_pre}

In this section, we will very briefly introduce the PLM framework of QSV proposed in \cite{Pallister_2018} and its subsequent developments. 
The task of quantum state verification is to determine if a quantum device produces a particular target state accurately, and a verification protocol provides us an explicit method to assess whether the output of this device is the target state rather than other states. More explicitly, given a target state $\ket{\psi}$, a quantum device produces a sequence of states $\rho_1,\cdots,\rho_N$, which may be i.i.d. or adversarial, then we need to accept one of the following statements: (1) the device produces the target state exactly, i.e., $\rho_i=\ket{\psi}\bra{\psi}$ for all $1\le i\le N$, or (2) the device can't produce the target state exactly, meaning $\bra{\psi}\rho_i\ket{\psi}\le 1-\epsilon$ with $0<\epsilon<1$ for all $i$.

\begin{figure}
    \centering
    \includegraphics[width=0.55\linewidth]{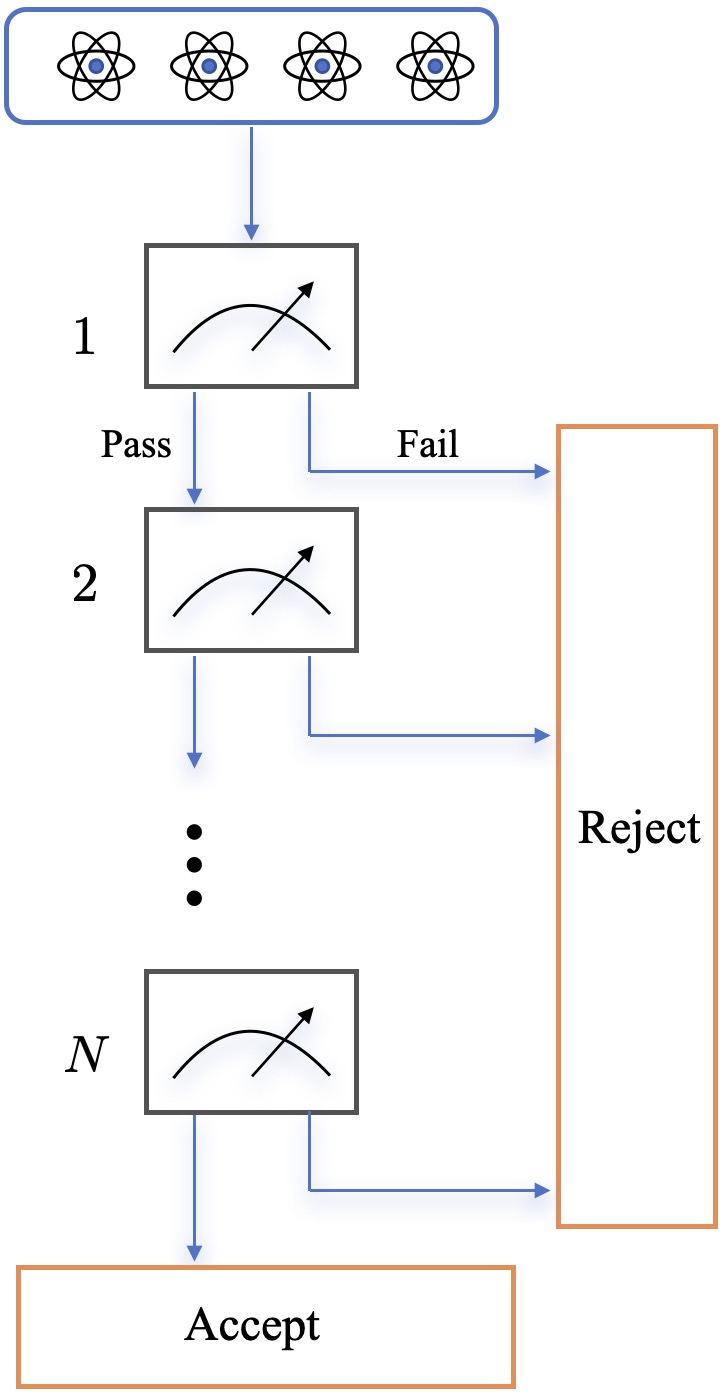}
    \caption{The illustration of the procedure of the PLM framework. A quantum device is accepted until it passes $N$ tests successively. }
    \label{fig:plm}
\end{figure}

\subsection{The PLM framework of QSV}
In the PLM framework of QSV, the null hypothesis $H_0$ is: the states produced by the quantum device satisfy $\bra{\psi}\rho_i\ket{\psi}\le 1-\epsilon$ for all $i$, and the alternative hypothesis $H_1$ is: the device produces the target state exactly, i.e., $\rho_i=\ket{\psi}\bra{\psi}$ for all $i$.

As illustrated in Protocol \ref{alg:state_verification}, an explicit protocol in the PLM framework of QSV will proceed as the following. (1) When $i$ runs from $1$ to $N$, the verifier draws $\{E_j, \md{1}-E_j\}$ from a set $\mc{E}$ of binary-outcome projective measurements with probability $p_j$, then applies it to $\rho_i$ with the outcomes of $E_j, \md{1}-E_j$ corresponding to "Pass" and "Fail" respectively. 
(2) If a "Fail" is obtained for some $i$, then the protocol ends, and we accept the null hypothesis $H_0$ and "reject" the quantum device. Otherwise, $H_1$ is accepted, and we say the device is "accepted", meaning the device produces the target state $\ket{\psi}$.
The procedure of the PLM framework is illustrated in Figure \ref{fig:plm}.
\floatname{algorithm}{Protocol}
\begin{algorithm}[H] 
  \caption{PLM framework of QSV}
  \label{alg:state_verification}
   \begin{algorithmic}[1]
   \For{$i = 1$ to $N$}
   \State Apply two-outcome measurement $E_i \in \mathcal{E}$ to $\rho_i$, where the outcomes are associated with ``Pass'' and ``Fail''.
   \If{``Fail'' is returned}
   \State Output ``Reject''
   \EndIf
   \EndFor
   \State Output ``Accept''
   \end{algorithmic} 
\end{algorithm}


There is a constraint on the test operator $E_j$: $E_j\ket{\psi}=\ket{\psi}$ for all $j$, which guarantees the target state always passes the test.
As for the produced state $\rho_i$, the probability of $\rho_i$ passing the test is
\begin{align}
    \te{Pr}[\text{"Pass"}]=\sum_{j}p_j\Tr(E_j\rho_i)=\Tr(\Omega \rho_i),
\end{align}
where $\Omega:=\sum_{j} p_jE_j$ is called the \textit{strategy operator} and satisfies $0\le \Omega \le \md{1}$ and $\Omega\ket{\psi}=\ket{\psi}$. Since $\rho_i$ satisfies $\bra{\psi}\rho_i\ket{\psi}\le 1-\epsilon$, the worst-case probability of $\rho_i$ passing the test is
\begin{align}
    \max_{\rho,\bra{\psi}\rho_i\ket{\psi}\le 1-\epsilon} \Tr(\Omega \rho)=1-\nu(\Omega)\epsilon,
\end{align}
where $\nu(\Omega)$ is the spectral gap of strategy operator $\Omega$. If we assume the eigenvalues of $\Omega$ are in the non-increasing order as $1=\lambda_1(\Omega)\ge \lambda_2(\Omega) \ge \cdots \ge \lambda_d(\Omega)\ge 0$ with $d$ being the dimension of the Hilbert space, then $\nu(\Omega):=1-\lambda_2(\Omega)$.

On the assumption of $E_j\ket{\psi}=\ket{\psi}$, the Type \RNum{2} error, which is the probability of accepting $H_0$ given $H_1$ is true, is always $0$. Therefore, there is a unique optimal acceptance set or decision rule in the view of the classical hypothesis testing \cite{cover2006elements, hogg2013introduction}, which is just the rule shown in Protocol \ref{alg:state_verification}. 
The Type \RNum{1} error, the probability of accepting $H_1$ given $H_0$ is true, is just the probability of all $N$ states $\rho_1,\cdots,\rho_N$ passing the test,
\begin{align}
    \te{Pr}[\text{"Accept"}]= \prod_{i=1}^N \Tr(\Omega \rho_i)\le [1-\nu(\Omega)\epsilon]^N.
\end{align}
Then by requiring the probability $\te{Pr}[\text{"Accept"}]$ being smaller than a positive constant $\delta$,
we can get the lower bound of the sample complexity,
\begin{align}
    N\ge \frac{\ln \delta}{\ln [1-\nu(\Omega)\epsilon]} \approx -\frac{\ln \delta}{\nu(\Omega)\epsilon}, \label{basic_complexity}
\end{align}
i.e., the number of $\rho_i$'s to verify if the device accurately produces the target state in the fidelity precision $\epsilon$ with confidence $1-\delta$. Eq.\eref{basic_complexity} tells us that the sample complexity increases when the spectral gap $\nu(\Omega)$ decreases, then for a specific target state, our goal is to design a strategy operator $\Omega$ with the spectral gap being as large as possible. If the spectral gap is maximal, the corresponding strategy operator $\Omega$ or protocol is \emph{optimal}. If the quantum device produces $n$-qubit states, we also say a protocol is \emph{efficient} if the inverse of the corresponding spectral gap $\nu(\Omega)$ is upper bounded by a polynomial of $n$.


\subsection{QSV with LOCC measurements}

There have also been many improvements on the PLM framework of QSV, amongst which is the substitution of the local measurements by the measurements composed of local operations and classical communications (LOCC) \cite{Yu_2019,Wang_2019,Liu_2019, Yu_2022, Li_2021}.
In \cite{Pallister_2018}, the test operator $E_i$'s are constructed by local projective measurements of each party of the system and there are no classical communications among different parties. 
If the classical communications among the parties are allowed, there are more types of test operators can be constructed. Taking the one-way LOCC strategy operator $\Omega^{\rightarrow}$ of a bipartite system as an example, Alice performs a local measurement $\{M_{a|j}\}$ with probability $p_j$ on her subsystem, then Alice sends $j$ and the measurement outcome $a$ to Bob, and then Bob makes the measurement $\{K_{a|j}, \md{1}-K_{a|j}\}$ on his subsystem based on $j, a$. 
So $\Omega^{\rightarrow}$ has the following form,
\begin{align}
    \Omega^{\rightarrow}=\sum_{j}p_j E_j^{\rightarrow}, 
    \quad E_j^{\rightarrow}=\sum_{a} M_{a|i}\otimes K_{a|i}.
    \label{eq:locc_measurement}
\end{align}
The above procedure can also be generalized to two-way LOCC and many-round two-way LOCC. 
One special example of one-way LOCC strategy operators is proposed in \cite{li2025universalefficientquantumstate}. In \cite{li2025universalefficientquantumstate}, the verifier is given a $n$-qubit state produced by a quantum device and apply a randomized local Pauli measurements to $n-1$ qubits randomly. After obtaining the measurement outcomes $\bs{s}=(s_1,\cdots,s_{n-1})$, the verifier applies a projective measurement to the remaining qubit, where the projective measurement is just the reduced state of $\ket{\psi}$ with respect to the outcome $\bs{s}$. Astonishingly, such simple protocol is proved to be valid for many types of states numerically or analytically in \cite{li2025universalefficientquantumstate}.

\section{The shadow overlap protocol and hypothesis testing}

\label{sec_sop}

\subsection{The shadow overlap protocol}

Parallel to advancements in the PLM framework for QSV, Hsin-Yuan Huang et al. recently proposed a new general protocol termed \textit{shadow overlap protocol} (SOP) in \cite{10756060_huang}, which certifies nearly all Haar random pure states. This protocol utilizes the classical shadow method \cite{huang2020predicting} and employs a classical query model $\mc{Q}(\psi)$, which is a classical description of the target state $\ket{\psi}$ and stored in a classical computer. 
At first glance, the SOP appears very different from the PLM framework, but it also shares many similarities with the PLM framework, particularly with the previous mentioned protocol in \cite{li2025universalefficientquantumstate}. Note that the SOP is designed to handle only i.i.d. states produced by the quantum device.
Here, we briefly review the SOP and reformulate it by using the terminologies of \emph{hypothesis testing} to unify its description and the PLM framework and facilitate comparison.

\begin{figure}
    \centering
    \includegraphics[width=1.0\linewidth]{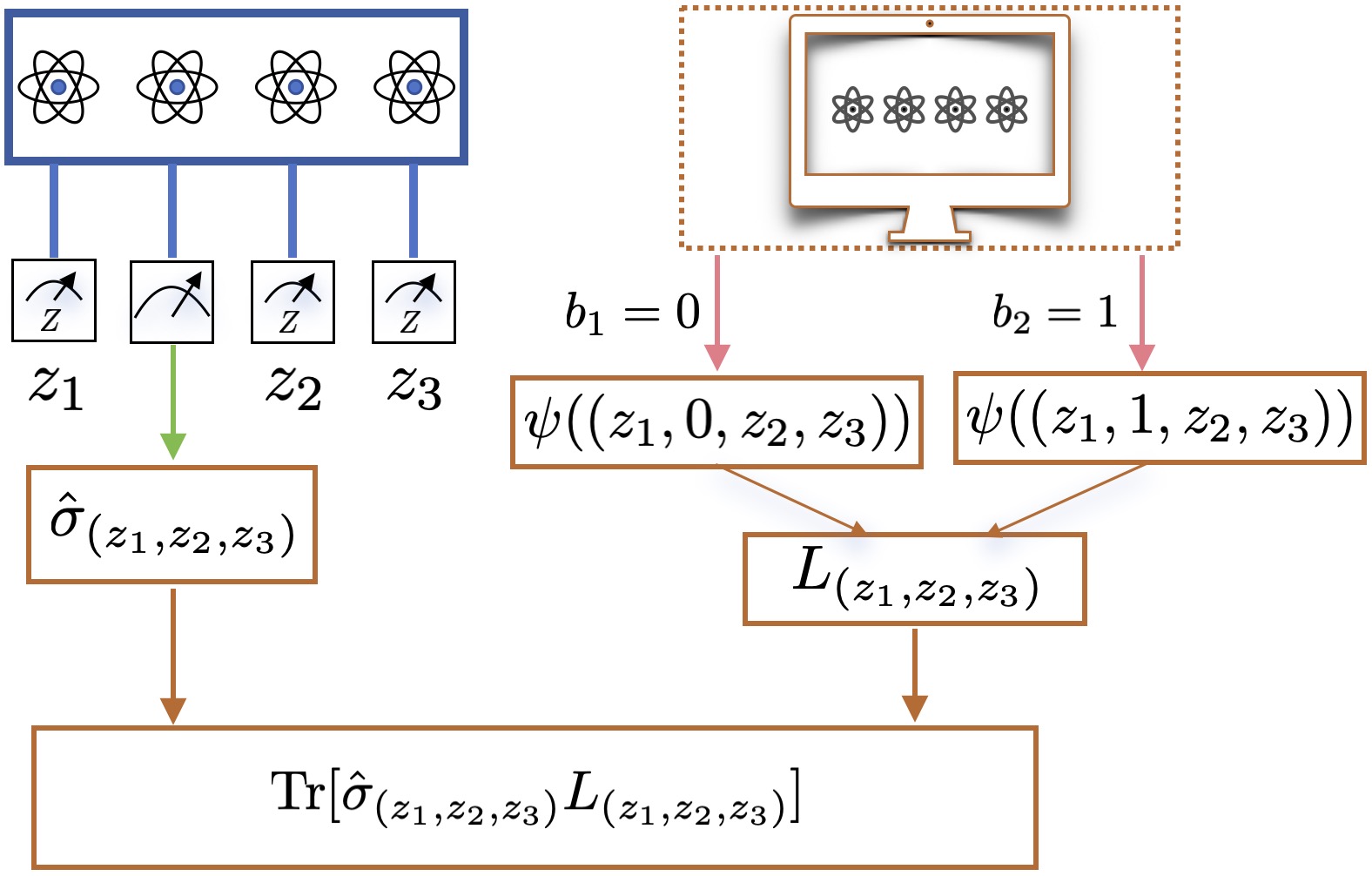}
    \caption{Illustration of the SOP procedure for $n=4, \ell=1$. The left blue rectangle represents the $4$-qubit state $\rho$ produced by a quantum device, while the right dashed rectangle represents the query model of the target state $\ket{\psi}$ stored in a classical computer.  On the left side, the $2$nd qubit is sampled in the first step and Pauli-$Z$ measurements are applied to the remaining three qubits, yielding measurement outcomes $(z_1,z_3,z_4)$. The classical shadow of the $2$nd qubit's state, $\hat{\sigma}_{(z_1,z_3,z_4)}$, is then constructed from the outcome of the randomized Pauli measurement. On the right side, given the measurement outcomes $(z_1,z_3,z_4)$, the classical model is queried to obtain the amplitudes corresponding to $\tilde{\bs{z}}_1=(z_1,0,z_3,z_4)$ and $\tilde{\bs{z}}_2=(z_1,1,z_3,z_4)$. After constructing the projector $L_{(z_1,z_3,z_4)}$, the shadow overlap is computed classically.}
    \label{fig:sop}
\end{figure}

Suppose the $n$-qubit target state is $\ket{\psi}$ with its classical query model $\mc{Q}(\psi)$ stored in a classical computer and a quantum device produces states $\rho$'s which are i.i.d., the procedure of each run of the level-$\ell$ SOP comprises of two stages. 
\begin{itemize}
\item The first stage mainly consists of the measurements on a quantum state $\rho$ as following:
\begin{enumerate}
	\item Label $n$ qubits of state $\rho$ by $\{1,\cdots,n\}$, uniformly sample a random subset with at most $\ell$ qubits, denoted as $K=\{k_1,\cdots,k_r\}$ with $1\le r\le \ell$.
   
	\item Perform Pauli-$Z$ measurements to the remaining $n-r$ qubits, obtain the outcomes $\bs{z}_{\bar{K}}=(z_{1},\cdots,z_{n-r})\in \{0,1\}^{n-r}$ with $\bar{K}$ being the complement of $K$.
	
	\item Uniformly select a Pauli-$X$, $Y$ or $Z$ measurement and apply it to each qubit in $K$, then compute the classical shadow 
	\begin{align}
		\hat{\sigma}_{\bs{z}_{\bar{K}}} = \otimes_{i=1}^{r}(3\ket{s_i}\bra{s_i}-I),
	\end{align}
	where $(s_1,\cdots,s_r)$ are the corresponding measurement outcomes.

\end{enumerate}    

\item The second stage consists of classical post-processing solely.
\begin{enumerate}
	\item Query the classical model $\mc{Q}(\psi)$ to obtain all the amplitudes $\psi(\tilde{\bs{z}})=\braket{\tilde{\bs{z}}}{\psi}$ of $\ket{\psi}$ with $\tilde{\bs{z}}$ matching $\bs{z}_{\bar{K}}$ at $\bar{K}$.
	
	\item If $\tilde{\bs{z}}_1$ of amplitude $\psi(\tilde{\bs{z}}_1)$ and $\tilde{\bs{z}}_2$ of amplitude $\psi(\tilde{\bs{z}}_2)$ are completely different at $K$,  construct the normalized state
    \begin{align}
        \ket{\phi(\bs{b}_1,\bs{b}_2)}\equiv \frac{\psi(\tilde{\bs{z}}_1)\ket{\bs{b}_1}+\psi(\tilde{\bs{z}}_2)\ket{\bs{b}_2} }{\sqrt{|\psi(\tilde{\bs{z}}_1)|^2+|\psi(\tilde{\bs{z}}_2)|^2 }}
    \end{align}
    where two bit strings $\bs{b}_1, \bs{b}_2\in \{0,1\}^r$ are the components of $\tilde{\bs{z}}_1, \tilde{\bs{z}}_2$ at $K$.

    \item Combine all $\ket{\phi(\bs{b}_1,\bs{b}_2)}$ into a projector
    \begin{align}
    L_{\bs{z}_{\bar{K}}}=\sum_{\bs{b}_1,\bs{b}_2}  \ket{\phi(\bs{b}_1,\bs{b}_2)}  \bra{\phi(\bs{b}_1,\bs{b}_2)},
    \label{operator_L_z}
    \end{align}
    with the Hamming distance of $\bs{b}_1,\bs{b}_2$ being $r$.
    
    \item Classically compute the shadow overlap $\hat{\omega}=\Tr(\hat{\sigma}_{\bs{z}_{\bar{K}}} L_{\bs{z}_{\bar{K}}})$.
\end{enumerate} 

\item Repeat the previous two stages $N$ times, compute the average of $N$ shadow overlaps
\begin{align}
    \hat{\omega}_{\te{a}}=\frac{1}{N}\sum_{i=1}^N \hat{\omega}_i.
\end{align}

\end{itemize}
Note that the construction of classical shadow can also be classified into the stage of classical post-processing. 
The expectation value of projector $L_{\bs{z}_{\bar{K}}}$ with respect to the random sampling of $K$ and $\bs{z}_{\bar{K}}$ is
\begin{align}
    L:=&\mb{E}_{K, \bs{z}_{\bar{K}}}[L_{\bs{z}_{\bar{K}}}] 
    =\frac{1}{\sum_{i=1}^{\ell}\binom{n}{i}} \sum_{r=1}^{\ell} \notag\\
    &\sum_{K}\sum_{\bs{z}_{\bar{K}}\in \{0,1\}^{n-r}} \ket{\bs{z}_{\bar{K}}}\bra{\bs{z}_{\bar{K}}}\otimes L_{\bs{z}_{\bar{K}}}.
    \label{operator_L}
\end{align}
The operator $L$ satisfies $L\ket{\psi}=\ket{\psi}$, and has the relationship with the shadow overlap as, 
\begin{align}
    \mb{E}[\hat{\omega}]
    =\mb{E}_{K, \bs{z}_{\bar{K}}}\mb{E}_{\te{shadows}}[\hat{\omega}]
    =\Tr(\rho L).
\end{align}
So the average $\hat{\omega}_{\te{a}}$ is a statistical estimator of $\Tr(\rho L)$.

\subsection{The hypothesis testing}
\label{sec_sop_hypothesis}

In this subsection, we reformulate the SOP in the form of hypothesis testing and discuss two types of error and the sample complexity in details. Some results can be directly applied to our protocol later. 

First, we choose the same null and alternative hypotheses as those in the PLM framework. Before considering two type of errors, we need to select a decision rule, which determines which hypothesis is accepted according to the value of estimator $\hat{\omega}_{\te{a}}$ \cite{cover2006elements, hogg2013introduction}. 
If $H_1$ is true, $\Tr(L\rho)=1$, then as an estimator of $\Tr(L\rho)$, $\hat{\omega}_{a}$ fluctuates around $1$. While if $H_0$ is true, $\Tr(L\rho)\le 1-\nu(L)\epsilon$ with $\nu(L)$ being the spectral gap of $L$, then $\hat{\omega}_{a}$ fluctuates around $\Tr(L\rho)$ but is less than $1-\nu(L)\epsilon$ probably. Hence, we tend to reject $H_0$ and accept $H_1$ if $\hat{\omega}_{a} > 1-\nu(L)\epsilon$.
So we set the \textit{decision rule} as: 
\begin{align}
    \left\{ \begin{array}{cc}
        \te{reject}~H_0 & \te{if}~\hat{\omega}_{a}> t_0, \\
        \te{reject}~H_1 & \te{if}~\hat{\omega}_{a}\le t_0,
    \end{array} \right. \label{decision_rule}
\end{align}
with $1-\nu(L)\epsilon < t_0< 1$. 
In fact, different values of $t_0$ generate different tests, but it will be much simpler when setting $t_0=1-\nu(L)\epsilon/2$ as will be seen.

Next we try to compute the Type \RNum{1} and Type \RNum{2} errors.
The Type \RNum{1} error is
\begin{align}
    \te{Pr}[\te{reject}~ H_0|H_0]=& \te{Pr}\left\{\hat{\omega}_{a} \ge t_0 \right\} \notag\\
    =& \te{Pr}\left\{ \hat{\omega}_{a}\ge 1-\nu(L)\epsilon + t_0 - 1+\nu(L)\epsilon \right\} \notag\\
    \le & \te{Pr}\left\{ \hat{\omega}_{a}\ge \Tr(L\rho)+t_0 - 1+\nu(L)\epsilon \right\} \notag\\
    = & \te{Pr}\left\{ \hat{\omega}_{a} - \mb{E}[\hat{\omega}_{a}]\ge t_0 - 1+\nu(L)\epsilon \right\} \notag\\
    \le& \exp\left\{ -\frac{2N [t_0 - 1+\nu(L)\epsilon]^2}{(b-a)^2} \right\}, \label{type_1_error}
\end{align} 
where in the third line, we used the fact $\Tr(L\rho)\le 1-\nu(L)\epsilon$, and in the last line, we used the result of the application of the Hoeffding's inequality to $\hat{\omega}_a=\frac{1}{N}\sum_{i=1}^N \hat{\omega}_i$ as
\begin{align}
    \te{Pr}\left\{\hat{\omega}_a-\mb{E}[\hat{\omega}_a] \ge t  \right\}
    \le \exp\left\{ -\frac{2N t^2}{(b-a)^2} \right\}, \label{hoeffding}
\end{align}
with $t>0$ and $a\le \hat{\omega}_i \le b$ for all $i$. 
While the Type \RNum{2} error is
\begin{align}
    \te{Pr}[\te{reject}~ H_1|H_1]
    =& \te{Pr}\left\{\hat{\omega}_{a}\le t_0 \right\} \notag\\
    = & \te{Pr}\left\{ \hat{\omega}_{a} - \mb{E}[\hat{\omega}_{a}]\le t_0 -1 \right\} \notag\\
    \le& \exp\left\{ -\frac{2N(1-t_0)^2}{(b-a)^2} \right\}, \label{type_2_error}
\end{align}
where we used the result $\mb{E}[\hat{\omega}_{e}]=\Tr(L\rho)=1$ for $H_1$ and the result of the Hoeffding inequality 
\begin{align}
    \te{Pr}\left\{\hat{\omega}_a-\mb{E}[\hat{\omega}_a] \le -t  \right\}
    \le \exp\left\{ -\frac{2N t^2}{(b-a)^2} \right\}.
\end{align}

The upper bounds of Type \RNum{1} and Type \RNum{2} errors in Eq.\eref{type_1_error} and Eq.\eref{type_2_error} show that the speeds of both types of error approaching $0$ are exponentially fast with respect to $N$. But if $N$ is fixed, when $t_0$ approaches $1$, the Type \RNum{2} error becomes larger and the Type \RNum{1} error becomes smaller. It's just the reverse case when $t_0$ approaches $1-\nu(L)\epsilon$.
There is a tradeoff between the two errors, so we have to let one error be smaller than a small positive constant $\chi$ and compute the minimum of another error, just as what we should do in the standard hypothesis testing \cite{cover2006elements, hogg2013introduction}. After requiring the second minimal error to be smaller than $\delta$, then we can compute the sample complexity as shown in the following theorem. 
\begin{theorem}
\label{theorem_sop}
    Suppose that the Type \RNum{2} error $\beta:=\rm{Pr}[\rm{reject}~ H_1|H_1]$ is smaller than $\chi$ with $0<\chi<1$, then the Type \RNum{1} error $\alpha:=\rm{Pr}[\rm{reject}~ H_0|H_0]$ takes the optimal (minimal) value $\alpha^*$ when 
    \begin{align}
        t_0=1-(b-a)\sqrt{\frac{\ln \chi^{-1}}{2N}}.
    \end{align}
    After requiring the $\alpha^*\le \delta$, the sample complexity $N$ is
    \begin{align}
        N\ge \frac{(b-a)^2}{2\nu^2(L)\epsilon^2} \left(\sqrt{\ln \chi^{-1}} + \sqrt{\ln \delta^{-1}} \right)^2. \label{sample_complexity_general}
    \end{align}
\end{theorem}
\begin{proof}
    Fixing $\beta$ being smaller than $\chi$, i.e., the upper bound of $\beta$ is not larger than $\chi$, we can easily compute the condition about $t_0$ as
    \begin{align}
        t_0 \le 1-(b-a)\sqrt{\frac{\ln \chi^{-1}}{2N}}.
    \end{align}
    As $\alpha$ decreases when $t_0$ increases, so $\alpha$ is minimal when $t_0=1-(b-a)\sqrt{\ln \chi^{-1}/(2N)}$. Then we expect the Type \RNum{1} error being smaller than $\delta$, i.e., 
    \begin{align}
        \exp\left\{ -\frac{2N}{(b-a)^2} \left[ \nu(L)\epsilon - (b-a)\sqrt{\frac{\ln \chi^{-1}}{2N}}\right]^2 \right\} \le \delta,
        \notag
    \end{align} 
    and we can compute the sample complexity as
    \begin{align}
        N \ge \frac{(b-a)^2}{2\nu^2(L)\epsilon^2} \left(\sqrt{\ln \chi^{-1}} + \sqrt{\ln \delta^{-1}} \right)^2.
    \end{align}
\end{proof}
In fact, we can also fix the Type \RNum{1} error $\alpha$ and compute the optimal Type \RNum{2} error $\beta$, but we get the same sample complexity as Eq.\eref{sample_complexity_general}.
If we require both types of error to be smaller than one small constant like $\delta$, the situation is much simplified, and $t_0$ takes the value $1-\nu(L)\epsilon/2$,
then we can compute the sample complexity as
\begin{align}
    N \ge 2(b-a)^2 \frac{\ln \delta^{-1} }{\nu^2(L)\epsilon^2}.
    \label{sample_complexity_simple}
\end{align}
In the following, we will always let $t_0=1-\nu(L)\epsilon/2$.

Returning to the discussion of the SOP, the decision rule is: if the estimated shadow overlap $\hat{\omega}_{\te{e}} \ge 1-\epsilon\nu(L)/2$, 
then we conclude that the quantum device produces the target state.
According to Eq.\eref{operator_L}, the value of $\hat{\omega}$ is bounded as $0<\hat{\omega}<2^{2\ell-1}$ as proved in \cite{10756060_huang}, then we get the explicit sample complexity
\begin{align}
    N\ge 2^{4\ell-1}\frac{\log \delta^{-1} }{\epsilon^2\nu^2(L)}.
    \label{sample_complexity_sop}
\end{align}
Although the sample complexity $N$ scales exponentially with the level $\ell$, in practice the level $\ell$ can be set to a small number e.g., $1$ or $2$. In Fig. \ref{fig:averaged_spectral_gaps}, we consider the case where the target states are Haar random pure states and show the behaviors of the averaged spectral gap $\nu(L)$ with $n$ and $\ell$, and Fig. \ref{fig:histograms} illustrates the histograms of the spectral gaps for different $n$.
The differences between the above sample complexity with that of PLM framework will be discussed below.

\subsection{Comparison with the PLM framework}

There are many striking similarities between the SOP and the PLM framework. 
First, the operator $L$ defined in Eq.\eref{operator_L} has the same properties as the strategy operator $\Omega$ in the PLM framework:  $0\le L \le \md{1}$ and $L\ket{\psi}=\ket{\psi}$. 
Second, the procedure of randomized measurement in the SOP is very similar with that in \cite{li2025universalefficientquantumstate}, especially the case of $\ell=1$. 

However, there are also some fundamental differences between the SOP and the PLM framework.
First, the SOP assumes that the quantum device produces the i.i.d. states, but the $\rho$'s can be adversarial in the PLM framework. 
Secondly, the measurement applied in each run of the SOP is not a LOCC measurement but rather a local Pauli measurement, as there is no classical communication among $n$ parties before performing the Pauli measurements. This fundamental distinction sets it apart from the LOCC measurements described in Eq.\eref{eq:locc_measurement}. Although the Pauli measurements are performed on $n$ qubits in two steps during the first stage of the SOP, in a real experiment, they can instead be executed simultaneously, with the classical shadow computed in the second stage. 
Thirdly, although we have reformulated the SOP using hypothesis testing terminology, the resulting form differs from standard hypothesis testing as described in \cite{cover2006elements, hogg2013introduction}. The SOP more closely resembles parameter estimation, as its objective is to estimate $\Tr(L\rho)$.  
In the PLM framework, $\Tr(\rho\Omega)$ is the probability for $\rho$ passing a test and the decision rule is determined by the number of "pass". While in the SOP, we need to estimate $\Tr(\rho L)$ and make a decision based on the estimated value.
Fouthly, as shown in Eq.\eref{sample_complexity_sop}, the sample complexity of the SOP appears significantly larger than that of the PLM framework due to the exponential term $2^{4\ell-1}$. However, when $\ell$ is set to a small value, this term becomes negligible compared to other factors. 
Regarding other factors, the sample complexity $N$ of the SOP scales as $(\epsilon^2\nu^2(L))$, which is the consequence of the Hoeffding's inequality. While this scaling is less favorable, it enables the SOP to certify a broader class of states. In other words, there is a trade-off between the capability of uniformly verifying different types of state and the sample complexity of the protocol. 
 
In addition to its capability of certifying a broader class of states, particularly nearly all Haar random pure states, another advantage of the SOP is its reliance on local Pauli measurements. These measurements are significantly easier to implement experimentally compared to LOCC measurements.

\section{Directly partial shadow overlap protocol}
\label{sec_protocol}

\begin{figure}
    \centering
    \includegraphics[width=1.0\linewidth]{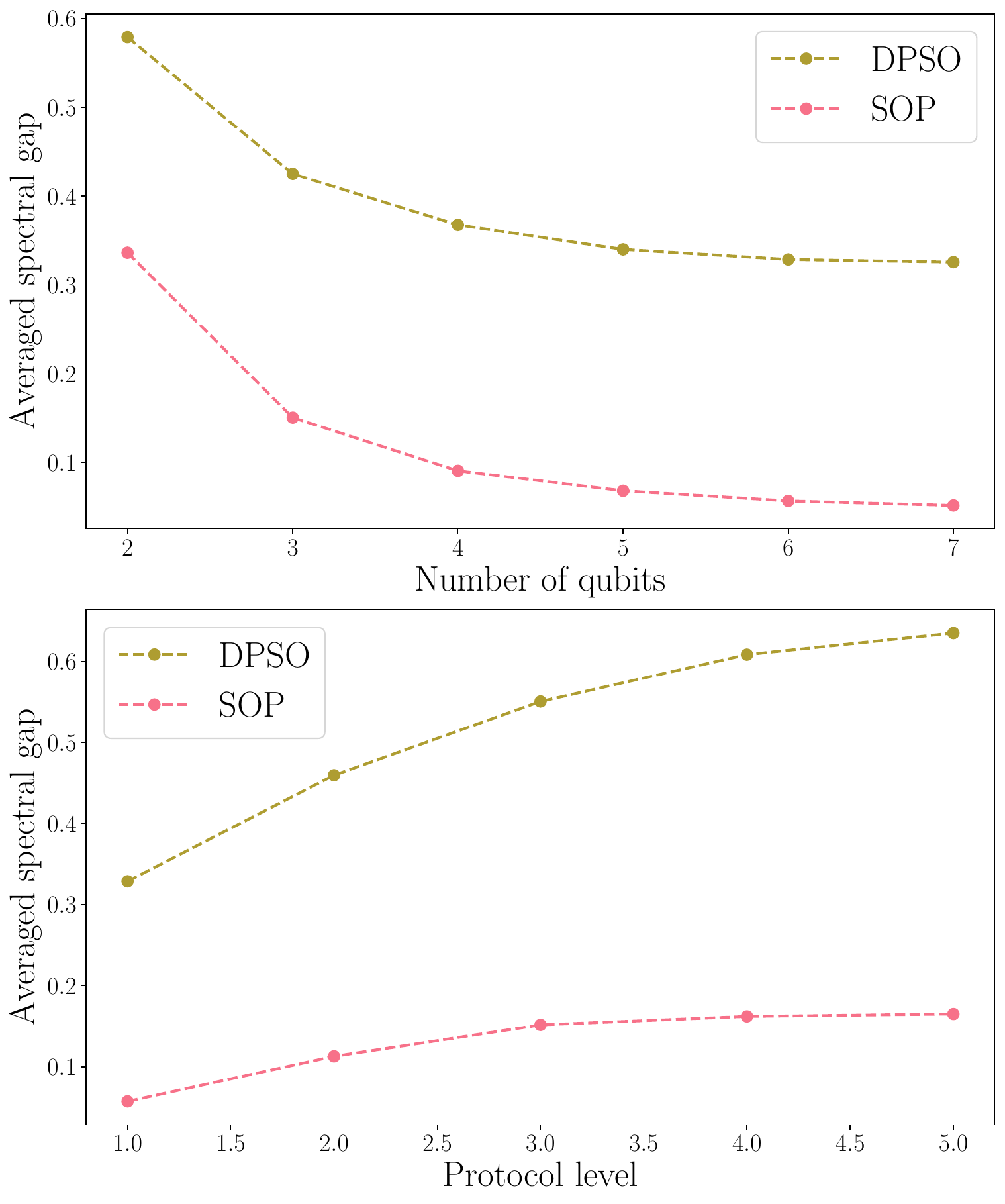}
    \caption{Illustration of the averaged spectral gaps of $\nu(L)$ and $\nu(\Omega)$ for the SOP and DPSO protocol. The target states are  Haar-random pure states, and the averaged spectral gaps are computed over $1000$ samples. The first figure shows that the averaged spectral gaps decreases as the number of qubits increases,  and that averaged $\nu(\Omega)$ is consistently larger than the averaged $\nu(L)$. The second figure illustrates the averaged spectral gaps increase with the protocol level, and the averaged spectral gap $\nu(\Omega)$ exceeds the averaged $\nu(L)$ across all levels. }
    \label{fig:averaged_spectral_gaps}
\end{figure}

We have discussed in detail the differences between the PLM framework and the SOP. In the SOP, only Pauli-Z measurements are applied to $n-r$ qubits, which prevents the direct application of the SOP to certain states, such as GHZ states. Additionally, the operator $L$ appears somewhat unnatural and complex. To address these limitations, we aim to generalize the protocol by incorporating other types of Pauli measurements and selecting a more natural operator than $L$. Inspired by measurements used in \cite{li2025universalefficientquantumstate}, we propose a new protocol, called as the \emph{directly partial shadow overlap} (DPSO) protocol, which follows the same core principles as the SOP and utilizes the classical shadow method.

\subsection{Description of the verification procedure}

In general, our DPSO protocol is similar with the SOP in \cite{10756060_huang} and both of them are based on the i.i.d. assumption for the states produced by the quantum device. As outlined in Protocol \ref{alg:my_protocol}, we need to run $N$ times of Algorithm \ref{alg:my_subalgorithm} to compute $N$ shadow overlaps $\hat{\omega}_a$. In Algorithm \ref{alg:my_subalgorithm}, the steps $1, 2, 3$ correspond to the first stage of the SOP, while the steps $4, 5$ correspond to the second stage of classical post-processing. The procedure of Algorithm \ref{alg:my_subalgorithm} is also illustrated in Figure \ref{fig:dpso}.

Compared with the SOP, there are several key changes in DPSO protocol. 
\begin{itemize} 
\item In the DPSO protocol, $r$ directly represents the protocol's level and there is no such parameter $\ell$. We also allow the set of $r$ qubits to be sampled following a general probability distribution. 
\item For the randomly selected $n-r$ qubits, we randomly apply Pauli-$X$, $Y$, or $Z$ measurements, rather than only the Pauli-$Z$ measurements. 
\item We also require a classical description of the target state $\ket{\psi}$ in order to directly compute its reduced state corresponding to the measurement outcomes $\bs{z}$, rather than construct a complicated operator $L_{z_{\bar{K}}}$ by querying a classical model of $\ket{\psi}$. 
\end{itemize} 
The last two differences are more important. The second modification enables us to handle a wider range of states with special structures, such as the GHZ states and stabilizer states. The third adjustment opens the possibility of constructing the optimal protocol, where the strategy operator achieves the optimal spectral gap. 
Hence, when the level $r=1$ and only Pauli-$Z$ measurements are applied, the DPSO protocol reduces to the SOP and is similar with the protocol in \cite{li2025universalefficientquantumstate}, then the DPSO protocol can also verify almost all Haar random pure states, just as the SOP \cite{10756060_huang}. 
Finally, it is important to emphasize that the measurement applied in each run of the DPSO protocol is not an LOCC measurement, but rather a randomized Pauli measurement. Therefore, the DPSO protocol should not be regarded as an adaptive version of the SOP.

\floatname{algorithm}{Algorithm}
\begin{algorithm}[H] 
  \caption{Estimate the overlap of two reduced states}
  \label{alg:my_subalgorithm}
\begin{algorithmic}[1]
   
   \State Label $n$ qubits of state $\rho$ by $[n]=\{1,\cdots,n\}$, randomly sample a subset of $n-r$ qubits $J=\{j_1,\cdots,j_{n-r}\}$ with probability $p_{J}$. The set of remaining qubits is denoted by $K=\{k_1,\cdots,k_r\}$.
   
   \State Randomly sample $\bs{l}=(l_1,\cdots,l_{n-r})$ from the set $\{1,2,3\}^{n-r}$ with probability $q_{\bs{l}}$, then apply the Pauli measurement $\sigma_{l_1}^{(j_1)},\cdots,\sigma_{l_{n-r}}^{(j_{n-r})}$ to the qubits $j_1,\cdots,j_{n-r}$ of state $\rho$, and obtain the outcomes $\bs{z}=(z_{1},\cdots,z_{n-r})$.
   
   \State For the set of qubits $K$ with reduced state $\zeta_{K,\bs{z}}$, apply random Pauli measurements and construct its classical shadow $\hat{\zeta}_{K,\bs{z}}$.
   
   \State Given the measurement outcome $\bs{z}$ for qubits $J$, compute the reduced state $\ket{\phi_{K, \bs{z}}}$ of $\ket{\psi}$ for qubits $K$ classically. 
 
   \State Compute the overlap of the classical shadow $\hat{\zeta}_{K,\bs{z}}$ and the reduced state $\ket{\phi}_{K, \bs{z}}$ classically, 
    \begin{align}
    \hat{\omega}=\te{Tr}\left( \hat{\zeta}_{K,\bs{z}} \ket{\phi_{K, \bs{z}}}\bra{\phi_{K, \bs{z}}} \right). \notag
    \end{align}

\end{algorithmic} 
\end{algorithm}

\floatname{algorithm}{Protocol}
\begin{algorithm}[H] 
  \caption{Level-$r$ shadow overlap protocol}
  \label{alg:my_protocol}
\begin{algorithmic}[1]
   \For{$i = 1$ to $N$}
        \State Run the Algorithm \ref{alg:my_subalgorithm}, obtain the shadow overlap $\hat{\omega}_i$
   \EndFor
   \State Compute the average $\hat{\omega}_{a}$ of $N$ shadow overlaps
   \If{$\hat{\omega}_a \ge 1 - \frac{1}{2}\nu(\Omega)\epsilon$}
        \State Output ``accept''
   \Else
        \State Output ``reject''
   \EndIf
\end{algorithmic} 
\end{algorithm}

\begin{figure}
    \centering
    \includegraphics[width=0.90\linewidth]{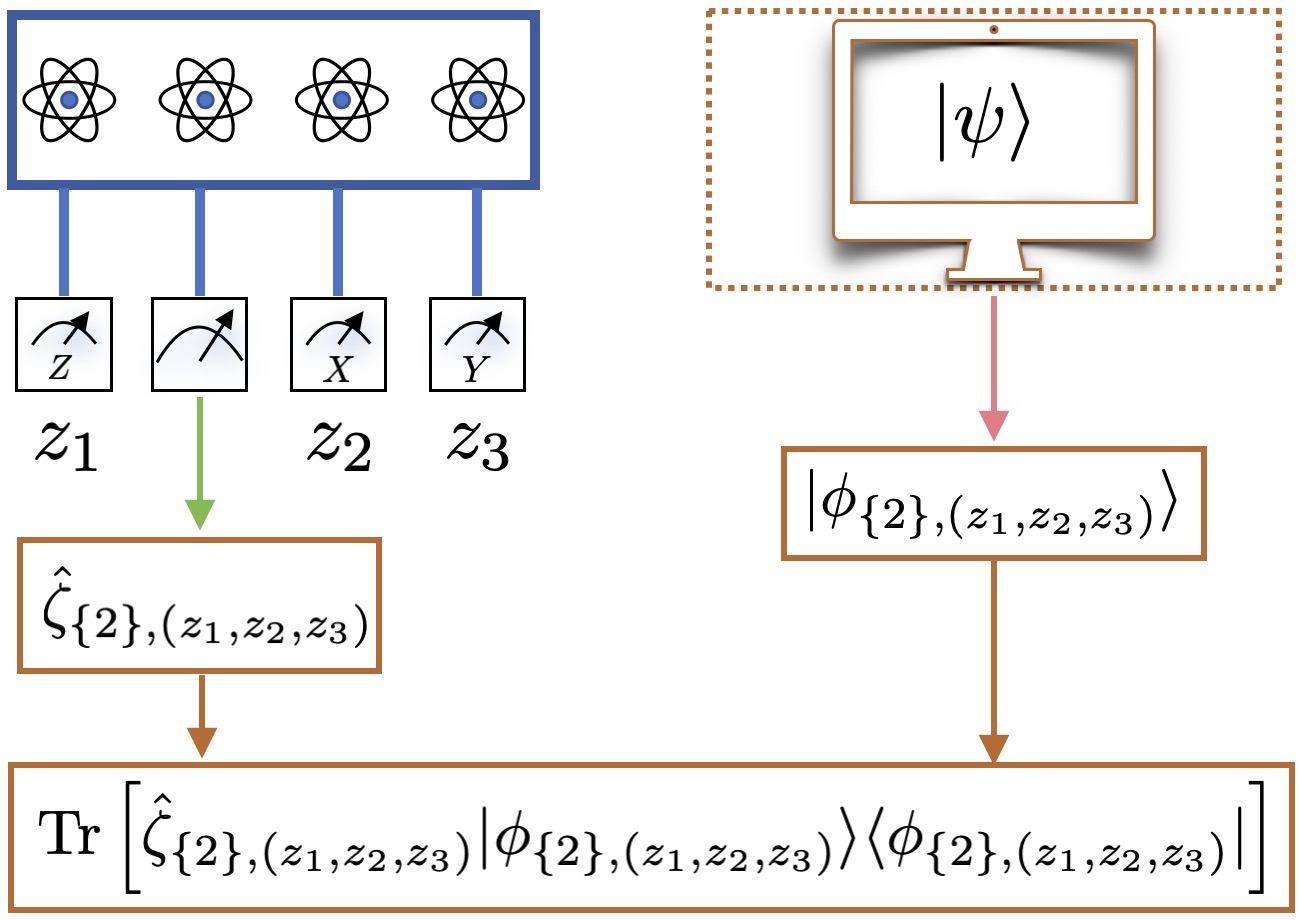}
    \caption{Illustration of the DPSO procedure for $n=4, \ell=1$. The left blue rectangle represents the $4$-qubit state $\rho$, while the right dashed rectangle represents the query model of the target state $\ket{\psi}$ in a classical computer.  On the left side, the $2$nd qubit is sampled in the first step and Pauli-$Z$ measurements are applied to the remaining three qubits, yielding measurement outcomes $(z_1,z_3,z_4)$. The classical shadow of the $2$nd qubit's state, $\hat{\sigma}_{(z_1,z_3,z_4)}$, is then constructed from the outcome of the randomized Pauli measurement. On the right side, given the measurement outcomes $(z_1,z_3,z_4)$, the classical model is queried to obtain the amplitudes corresponding to $\tilde{\bs{z}}_1=(z_1,0,z_3,z_4)$ and $\tilde{\bs{z}}_2=(z_1,1,z_3,z_4)$. After constructing the projector $L_{(z_1,z_3,z_4)}$, the shadow overlap is computed classically.}
    \label{fig:dpso}
\end{figure}

The intuition behind the DPSO protocol is that if the state $\rho$ produced by a quantum device equals to or approximates the target state $\ket{\psi}$, then after applying a random partial measurement to a random subset of the qubits, two ensembles of the reduced states will be the same or very similar. It's possible that after applying only Pauli-$Z$ measurements to a subset of qubits, the ensembles of reduced states may be close, even if two states are very different. To address this problem, all possible randomized partial measurements to random subsets of qubits are allowed.

Now we compute the expression of the strategy operator in the DPSO protocol. After applying the Pauli measurements $\sigma_{l_1}^{(j_1)}, \cdots, \sigma_{l_{n-r}}^{(j_{n-r})}$ to qubits $\{j_1,\cdots,j_{n-r}\}$ and obtaining the outcomes $\bs{z}=(z_1,\cdots,z_{n-r})\in \{\pm 1\}^{n-r}$, the reduced state of $\rho$ for qubits $K$ is
\begin{align}
\zeta_{K,\bs{z}}=\frac{\widetilde{\zeta}_{K,\bs{z}}}{\te{Tr}(\widetilde{\zeta}_{K,\bs{z}})},  \label{reduced_state_1}
\end{align}
where $\widetilde{\zeta}_{K,\bs{z}} = \bra{\bs{z}} \rho \ket{\bs{z}}$ and $\te{Tr}(\widetilde{\zeta}_{K,\bs{z}})$ is the probability of obtaining outcomes $\bs{z}$.
For qubits $K$, we apply random Pauli measurements and obtain the outcomes $\bs{s}=(s_1,\cdots,s_r)$, then we can construct the classical shadow of $\zeta_{K,\bs{z}}$ as
\begin{align}
    \hat{\zeta}_{K,\bs{z}}=\otimes_{a=1}^r(3\ket{s_a}\bra{s_a}-I),
\end{align}
where $\ket{s_a}$ is the eigenstate corresponding to the outcome $s_a$ of the Pauli measurement applied to qubit $k_a\in K$.
By the definition of the classical shadow \cite{huang2020predicting}, there exists the equality 
\begin{align}
    \mb{E}_{\te{sh}}[\hat{\zeta}_{K,\bs{z}}]=\zeta_{K,\bs{z}}.
\end{align}
Corresponding to the random Pauli measurements $(J, \bs{l})$ and the outcomes $\bs{z}$, the reduce state of the target state $\ket{\psi}$ for the qubits $K$ is
\begin{align}
\label{post_measure_psi}
\ket{\phi_{K,\bs{z}}} = \frac{\ket{\widetilde{\phi}_{K,\bs{z}}}}{\|\ket{\widetilde{\phi}_{K,\bs{z}}} \|},
\end{align}
with $\ket{\widetilde{\phi}_{K,\bs{z}}} := \braket{\bs{z}}{\psi}$. 
The computation of $\ket{\phi_{K,\bs{z}}}$ can be efficient for certain classical descriptions of the target state, such as when $\ket{\psi}$ is represented as a matrix product state, or under a more powerful query model than that considered in \cite{huang2020predicting}. Further discussion of this topic is provided in Appendix \ref{appendix_classical_target}.
Hence, the shadow expectation of $\hat{\omega}$ is
\begin{align}
\mb{E}_{\te{sh}}[\hat{\omega}] 
=& \te{Tr}_{K} \left( \mb{E}_{\te{sh}}[\hat{\zeta}_{K,\bs{z}}] \ket{\phi_{K, \bs{z}}}\bra{\phi_{K, \bs{z}}} \right) \notag\\
=&\frac{1}{\te{Tr}(\widetilde{\zeta}_{K,\bs{z}})} 
\te{Tr}_{K\cup J} \left( \rho \ket{\bs{z}}\bra{\bs{z}} \otimes \ket{\phi_{K, \bs{z}}}\bra{\phi_{K, \bs{z}}} \right) \notag\\
=&\frac{1}{\te{Tr}(\widetilde{\zeta}_{K,\bs{z}})} \te{Tr} \left( \rho \Omega_{K,\bs{l},\bs{z}} \right), \label{expectation_overlap}
\end{align}
where 
\begin{align}
    \Omega_{K,\bs{l},\bs{z}}:= \ket{\bs{z}}\bra{\bs{z}} \otimes \ket{\phi_{K, \bs{z}}}\bra{\phi_{K, \bs{z}}} 
    \label{sub_omega}
\end{align}
is a projector. 
Hence, $\mb{E}_{\te{sh}}[\hat{\omega}]$ has a clear physical meaning, proportional to the probability of projecting $\rho$ to the state $\ket{\bs{z}} \otimes \ket{\phi_{K, \bs{z}}}$. 

We further consider the average of $\mb{E}_{\te{sh}}[\hat{\omega}]$ over different outcomes $\bs{z}$, 
\begin{align}
    \sum_{\bs{z}} \te{Tr}(\widetilde{\zeta}_{K,\bs{z}}) \mb{E}_{\te{sh}}[\hat{\omega}] 
    = \te{Tr} \left( \rho \Omega_{K,\bs{l}} \right),
\end{align}
where $\Omega_{K,\bs{l}}$ is the summation of $\Omega_{K,\bs{l},\bs{z}}$ over all possible outcomes $\bs{z}$,
\begin{align}
\Omega_{K,\bs{l}} 
=& \sum_{\bs{z}} \ket{\bs{z}}\bra{\bs{z}} \otimes \ket{\phi_{K, \bs{z}}}\bra{\phi_{K, \bs{z}}}.
\end{align} 
The operator $\Omega_{K,\bs{l}}$ is also a projector and plays a similar role as $E_i$ in the PLM framework, so we call $\Omega_{K,\bs{l}}$ as the \textit{test operator} too.
The test operator $\Omega_{K,\bs{l}}$ satisfies a nice property, 
\begin{align}
\Omega_{K,\bs{l}}\ket{\psi}=& \left( \sum_{\bs{z}} \ket{\bs{z}}\bra{\bs{z}} \otimes \ket{\phi_{K, \bs{z}}}\bra{\phi_{K, \bs{z}}}\right) \ket{\psi}\notag\\
=& \sum_{\bs{z}} \ket{\bs{z}}\otimes \ket{\phi_{K, \bs{z}}} \braket{\phi_{K, \bs{z}}}{\widetilde{\phi}_{K, \bs{z}}} \notag\\
=& \ket{\psi}. \label{eigenequation_target}
\end{align}
The expectation value of the shadow overlap $\hat{\omega}$, 
\begin{align}
\omega=\mb{E}[\hat{\omega}]
=&\sum_{J} p_{J} \sum_{\bs{l}} q_{\bs{l}} \sum_{\bs{z}} \te{Tr}(\widetilde{\zeta}_{K,\bs{z}}) \mb{E}_{\te{sh}}[\hat{\omega}]
=\te{Tr} \left(\rho \Omega \right), \notag
\end{align}
then the average $\hat{\omega}_{a}$ in the protocol \ref{alg:my_protocol} is actually an estimator of $\omega$.
In the above equation, we have defined the operator
\begin{align}
\Omega:=\sum_{K} \sum_{\bs{l}} p_{K}q_{\bs{l}} \Omega_{K, \bs{l}}
=& \sum_{K} \sum_{\bs{l}} \sum_{\bs{z}} p_{K}q_{\bs{l}} \Omega_{K,\bs{l},\bs{z}},
\label{our_strategy_operator}
\end{align}
where we have used the fact that $p_J=p_K$ for $J\cup K=\{1,\cdots,n\}$. It's easy to prove that the operator $\Omega$ satisfies the following two properties: 
\begin{align}
    0\le \Omega \le \md{1}, \quad \Omega\ket{\psi}=\ket{\psi}.
\end{align}
Hence, the operator $\Omega$ plays the same role in the DPSO protocol as the operator $L$ in the SOP and is called the \textit{strategy operator} of the DPSO protocol.

\begin{figure}
    \centering
    \includegraphics[width=1.0\linewidth]{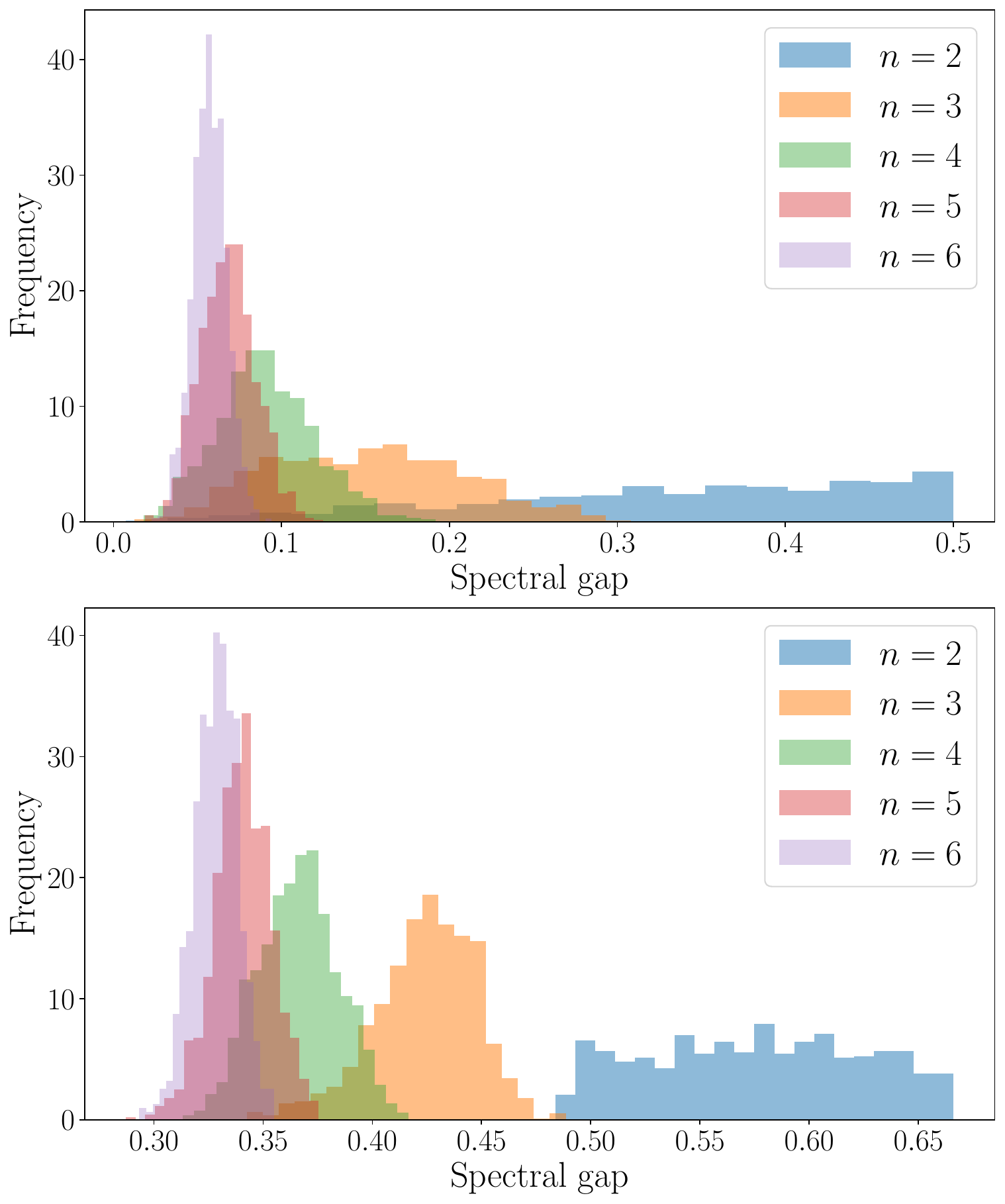}
    \caption{Illustration of the histograms of spectral gaps for the Haar random pure states with different number of qubits. The averaged spectral gaps are shown in Fig. \ref{fig:averaged_spectral_gaps}.
    Here, the first and second figures correspond to the SOP and DPSO protocols respectively. }
    \label{fig:histograms}
\end{figure}

\subsection{Sample complexity}

The DPSO protocol is similar with the SOP except the operator $L$ being replaced by the strategy operator $\Omega$ in Eq.\eref{our_strategy_operator}, then the results of Theorem \ref{theorem_sop} and Eq.\eref{sample_complexity_simple} in subsection \ref{sec_sop_hypothesis} can be applied to the DPSO protocol. For simplicity, we only consider the case where two types of error are equal. Then we obtain the following main result.
\begin{corollary}
Given a $n$-qubit target state $\ket{\psi}$, the DPSO Protocol \ref{alg:my_protocol} will accept the null hypothesis $H_0$
or the alternative hypothesis $H_1$ with probability at least $1-\delta$, if the number of copies of $\rho$ is
\begin{align}
    N=2^{2r+1}\frac{\ln \delta^{-1}}{\nu(\Omega)^2\epsilon^2},  \label{my_sample_complexity}
\end{align}
where $r$ is the level of the protocol, $\epsilon$ is the infidelity of $\rho$ with $\ket{\psi}$ and $\nu(\Omega)$ is the spectral gap of the strategy operator of the protocol.
\end{corollary}

\begin{proof}
From Eq.\eref{sample_complexity_simple}, we only need to compute the value range of the shadow overlap $\hat{\omega}$. It's easy to compute as
\begin{align}
|\hat{\omega}|
=& |\langle \hat{\zeta}_{K,\bs{z}}, \ket{\phi_{K,\bs{z}}}\bra{\phi_{K,\bs{z}}} \rangle_{\te{HS}}| \notag\\
\le& \| \hat{\zeta}_{K,\bs{z}}\|_{\infty} 
\|\ket{\phi_{K,\bs{z}}}\bra{\phi_{K,\bs{z}}}\|_1 \notag\\
\le& 2^r
\end{align}
where we have used the facts
\begin{align}
    \| \hat{\zeta}_{K,\bs{z}}\|_{\infty} 
    =\prod_{a=1}^r\| 3\ket{s_a}\bra{s_a}-I\|_{\infty}
    \le 2^r
\end{align}
and $\|\ket{\phi_{K,\bs{z}}}\bra{\phi_{K,\bs{z}}} \|_1\le 1$. As an overlap of two states, $\hat{\omega}\ge0$, then we get the range of value as $0\le \hat{\omega}\le 2^r$. Substituting $a=0, b=2^r$ into Eq.\eref{sample_complexity_simple}, we get the sample complexity.
\end{proof}

Compared with the sample complexity of the SOP in Eq.\eref{sample_complexity_sop}, the DPSO protocol's sample complexity has the same dependence on parameters $\delta, \epsilon$ and the spectral gap. But if two levels are equal, $r=\ell$, the sample complexity of the DPSO protocol is smaller than that of the SOP by a factor $2^{2\ell-2}$. Further, just as illustrated in Fig. \ref{fig:averaged_spectral_gaps} and Fig. \ref{fig:histograms}, the spectral gap $\nu(\Omega)$ may have larger value than $\nu(L)$ with high probability. 
When $r=\ell=1$, it's still probable that the DPSO protocol has a smaller sample complexity than the SOP, since the averaged spectral gap of DPSO is much smaller than that of SOP as illustrated in Fig.~\ref{fig:averaged_spectral_gaps}.
Moreover, the DPSO protocol can also be modified analogously to the modification of Protocol~$1$ in Ref.~\cite{huang2024certifying}, in order to achieve an improved dependence for the sample complexity as $\mathcal{O}(\epsilon^{-1}\nu(L)^{-1})$.

Given a specific target state $\ket{\psi}$, if the level $r$, $\delta$ and $\epsilon$ are fixed, the sample complexity only depends on the spectral gap of $\Omega$. Hence, to find the optimal strategy operator $\Omega$ in the DPSO protocol \ref{alg:my_protocol}, we need to solve the following optimization problem 
\begin{align}
    \max_{\Omega} \nu(\Omega)=\max_{\{p_J\}, \{q_{\bs{l}}\}} \nu(\Omega), \label{optimization_problem}
\end{align}
where $\Omega$ is in the form as shown in Eq.\eref{our_strategy_operator}.

\section{Stabilizer formalism}
\label{sec_stabilizer}

In order to illustrate the DPSO protocol, we apply it to the stabilizer state. We can show that our protocol can deal with the stabilizer state, including the GHZ state, more naturally than the SOP.
The problem of verification of the stabilizer state has been solved in the PLM framework elegantly \cite{Pallister_2018, Dangniam_2020}, and the protocols are optimal or near-optimal. So we don't seek to find an explicit optimal DPSO protocol for verifying a specific stabilizer state but provide a method to determine the optimal strategy operator.

\subsection{Pauli measurement and the test operator}
Now we formalize the DPSO protocol into a more formal form by using the symplectic representation of Pauli measurements, which will be very useful to deal with the stabilizer state. 
For a single qubit, the Pauli measurements $X, Z$ and $Y$ correspond to the vectors $(1;0), (0;1)$ and $(1;1)$ in the symplectic space $\mb{Z}_2^2$, while the trivial measurement, i.e., the identity or no measurement applied, corresponds to $(0;0)$.
Each \textit{composite} Pauli measurement on $n$ qubits is specified by a symplectic vector $\mu=(\mu^\mr{x};\mu^{\mr{z}})$ in the symplectic space $\mb{Z}_2^{2n}$, the Pauli measurement applied to qubit $j$ is $\mr{i}^{\mu_j^{\mr{x}}\mu_j^{\mr{z}}}X_j^{\mu_j^{\mr{x}}}Z_j^{\mu_j^{\mr{z}}}$.
The \textit{weight} of a composite Pauli measurement on $n$ qubits is the total number of non-trivial Pauli measurements applied to a single qubit, i.e., the total number of $j$ with $(\mu_j^{\mr{x}}, \mu_j^{\mr{z}})\ne (0,0)$. A composite Pauli measurement is \textit{complete} if its weight equals $n$. However, only the non-complete composite Pauli measurements are involved in our protocol.

Considering the DPSO protocol \ref{alg:my_protocol},
let's focus on a composite Pauli measurement $\mu \in \mb{Z}_2^{2n}$ with weight $t=n-r$ and assume the non-trivial single Pauli measurements are applied to the set of qubits $J=\{j_1, \cdots,j_t\}$. There are $t$ non-trivial Pauli operators like $\mathrm{i}^{\mu_{j_a}^{\mr{x}}\mu_{j_a}^{\mr{z}}}X_{j_a}^{\mu_{j_a}^{\mr{x}}}Z_{j_a}^{\mu_{j_a}^{\mr{z}}}$ with $a=1,\cdots,t$, then these Pauli operators generate a stabilizer group $\mc{T}_{\mu}$ with cardinality $|\mc{T}_{\mu}|=2^t$.  
Applying the composite Pauli measurement to a state $\rho$, the outcome of qubit $j_a$ is $(-1)^{v_a}$ with $v_a\in \{0, 1\}$, then the measurement outcomes of $t$ qubits correspond to a $t$-dimensional vector $\bs{v}\in \mb{Z}_2^t$. 
For each vector $\bs{v}\in \mb{Z}_2^k$, we can construct a new stabilizer groups $\mc{T}_{\mu, \bs{v}}$ generated by such Pauli operator $(-1)^{v_a} \mathrm{i}^{\mu_{j_a}^{\mr{x}}\mu_{j_a}^{\mr{z}}}X_{j_a}^{\mu_{j_a}^{\mr{x}}}Z_{j_a}^{\mu_{j_a}^{\mr{z}}}$ with $a=1,\cdots,t$, its corresponding projector is
\begin{align}
    \Pi_{\mu,\bs{v}}=& \bigotimes_{a=1}^{t} \frac{1}{2}\left[I+(-1)^{v_a} \mathrm{i}^{\mu_{j_a}^{\mr{x}}\mu_{j_a}^{\mr{z}}}X_{j_a}^{\mu_{j_a}^{\mr{x}}}Z_{j_a}^{\mu_{j_a}^{\mr{z}}}\right] \notag\\
    =& \frac{1}{2^t} \sum_{\bs{u}\in \mb{Z}_2^t} (-1)^{\bs{u}\cdot \bs{v}} W_{\mu}(\bs{u}), \label{outcome_projector}
\end{align}
where $W_{\mu}(\bs{u})$ is a Pauli operator like
\begin{align}
    W_{\mu}(\bs{u})=\bigotimes_{a=1}^t \left( \mathrm{i}^{\mu_{j_a}^{\mr{x}}\mu_{j_a}^{\mr{z}}}X_{j_a}^{\mu_{j_a}^{\mr{x}}}Z_{j_a}^{\mu_{j_a}^{\mr{z}}} \right)^{u_a}.
\end{align}
Note that $\mc{T}_{\mu,\bs{0}}=\mc{T}_{\mu}$ and $\Pi_{\mu,\bs{0}}=\Pi_{\mu}$.
Given a composite Pauli measurement $\mu\in \mb{Z}_2^{2n}$, the symplectic vector $\mu$ has included the information about the set of qubits $J$ and the Pauli measurements applied to these qubits, so we don't need to use the combination of $J$ or $K$ with $\bs{l}$, used in the last section, to label the measurement.

After applying the composite Pauli measurement $\mu$ to a state $\rho$ and obtaining the outcome $\bs{v}$, the reduced state is 
\begin{align}
\zeta_{\mu,\bs{v}}=
    \frac{\te{Tr}_{J}(\rho \Pi_{\mu,\bs{v}})}{\te{Tr}(\rho \Pi_{\mu,\bs{v}})}, \label{reduced_state_2}
\end{align}
which is just another form of Eq.\eref{reduced_state_1}.  
As for the strategy operator $\Omega$, we can first rewrite the test operator $\Omega_{K,\bs{l}}$ as 
\begin{align}
    \Omega_{\mu}=\sum_{\bs{v}\in \mb{Z}_2^t}\Pi_{\mu,\bs{v}}\otimes \frac{\te{Tr}_{J}(\rho \Pi_{\mu,\bs{v}})}{\te{Tr}(\rho \Pi_{\mu,\bs{v}})}.
    \label{test_operator_projector}
\end{align}
Note that the reduced states corresponding to some outcomes $\bs{v}$ may be $0$, then such terms vanish and disappear in Eq.\eref{test_operator_projector}, so we don't need to worry the denominator. The test operator $\Omega_{\mu}$ in Eq.\eref{test_operator_projector} has another more useful form, as summarized in the following theorem.

\begin{theorem} \label{theorem_test_operator}
Given a weight-$t$ composite Pauli measurement $\mu\in \mb{Z}_{2}^{2n}$, the non-trivial local Pauli measurements are applied to the set of qubits $J=\{j_1,\cdots,j_t\}$, then
the test operator $\Omega_{\mu}$ has the following form
\begin{align}
    \Omega_{\mu}= \frac{1}{2^t \Tr(\rho \Pi_{\mu})} \sum_{T\in \mc{T}_{\mu}} 
    T\otimes \mr{Tr}_{J}(\rho T),
    \label{test_operator_stabilizer}
\end{align}
where $T$ is an element of the stabilizer group $\mc{T}_{\mu}$ corresponding to $\mu$. 
\end{theorem}

\begin{proof}
Before deriving Eq.\eref{test_operator_stabilizer}, we first prove two results, which will be used in the proof. 
(1) If the reduced state $\te{Tr}_{J}(\rho \Pi_{\mu,\bs{v}})\ne 0$ in Eq.\eref{test_operator_projector}, then $\Tr(\rho \Pi_{\mu,\bs{v}})\ne 0$ and
\begin{align}
    \Tr(\rho \Pi_{\mu,\bs{v}})=\Tr(\rho \Pi_{\mu}).  
\end{align} 
To prove Eq.\eref{overlap_projector}, let's first assume that $\rho$ is a stabilizer state with stabilizer group $\mc{S}$, then according to the results in \cite{Dangniam_2020}, we know that
\begin{align}
    \Tr(\rho \Pi_{\mu,\bs{v}}) = \left\{
    \begin{array}{cc}
       \frac{|\mc{T}_{\mu}\cap \mc{S}|}{ |\mc{T}_{\mu}|} &  \te{if}~(-\mc{T}_{\mu})\cap \mc{S}= \emptyset, \\
       0  & \te{otherwise},
    \end{array} \right.  \label{overlap_projector}
\end{align}
for all $\bs{v}\in \mb{Z}_2^{t}$. Since $|\mc{T}_{\mu}\cap \mc{S}| /|\mc{T}_{\mu}|$ is independent with $\bs{v}$, it's easy to get the Eq.\eref{overlap_projector} for stabilizer states.   
When $\rho$ is a non-stabilizer state, it can be expanded as a linear combination of some stabilizer states as $\rho=\sum_i c_i \sigma_i$, then 
\begin{align}
    \te{tr}(\rho \Pi_{\mu,\bs{v}})
    =& \sum_i c_i \te{tr}(\sigma_i \Pi_{\mu})
    = \te{tr}(\rho \Pi_{\mu}),
\end{align}
where we have used the fact that $\te{tr}(\sigma_i \Pi_{\mu,\bs{v}}) = \te{tr}(\sigma_i \Pi_{\mu})$ when $\te{tr}(\sigma_i \Pi_{\mu,\bs{v}})\ne 0$.
(2) The second result is the equation, 
\begin{align}
    \sum_{\bs{v}\in \mb{Z}_2^t} (-1)^{(\bs{u}+\bs{w})\cdot \bs{v}} 
    = 2^t \delta_{\bs{u},\bs{w}}. \label{delta_function}
\end{align}
It's easy to see that if $\bs{u}= \bs{w}$, \begin{align}
    \sum_{\bs{v}\in \mb{Z}_2^{t}}(-1)^{\bs{0}\cdot \bs{v}}=\sum_{\bs{v}\in \mb{Z}_2^{t}}1=2^k,
\end{align} 
otherwise 
\begin{align}
    \sum_{\bs{v}\in \mb{Z}_2^{t}}(-1)^{(\bs{u}+\bs{w})\cdot \bs{v}}=0,
\end{align}
because for any $\bs{v}$ letting $(-1)^{(\bs{u}+\bs{w})\cdot \bs{v}}=1$, we can flip one bit of $\bs{v}$ and get a vector $\bs{v}'$ satisfying $(-1)^{(\bs{u}+\bs{w})\cdot \bs{v}'}=-1$.

To prove Eq.\eref{test_operator_stabilizer}, we directly substitute Eq.\eqref{outcome_projector} into Eq.\eqref{test_operator_projector} and derive
\begin{align}
    \Omega_{\mu}=& \frac{1}{2^{2t}}\sum_{\bs{v}\in \mb{Z}_2^t}  \sum_{\bs{u}\in \mb{Z}_2^t} \sum_{\bs{w}\in \mb{Z}_2^t} (-1)^{(\bs{u}+\bs{w})\cdot \bs{v}} \frac{1}{\te{Tr}(\rho \Pi_{\mu,\bs{v}})} \notag\\
    &   W_{\mu}(\bs{u}) \otimes \te{Tr}_{J}[\rho W_{\mu}(\bs{w})] \notag\\
    =&  \frac{1}{2^{2t}\te{Tr}(\rho \Pi_{\mu})} \sum_{\bs{u}\in \mb{Z}_2^t} \sum_{\bs{w}\in \mb{Z}_2^t}  \left( \sum_{\bs{v}\in \mb{Z}_2^t} (-1)^{(\bs{u}+\bs{w})\cdot \bs{v}} \right) \notag\\
    &W_{\mu}(\bs{u}) \otimes \te{Tr}_{J}[\rho W_{\mu}(\bs{w})]\notag\\
    =&  \frac{1}{2^{2t}\te{Tr}(\rho \Pi_{\mu})} \sum_{\bs{u}\in \mb{Z}_2^t} \sum_{\bs{w}\in \mb{Z}_2^t}  2^t \delta_{\bs{u},\bs{w}} W_{\mu}(\bs{u}) \otimes \te{Tr}_{J}[\rho W_{\mu}(\bs{w})] \notag\\
    =&  \frac{1}{2^{t}\te{Tr}(\rho \Pi_{\mu})} \sum_{\bs{w}\in \mb{Z}_2^t}  W_{\mu}(\bs{w}) \otimes \te{Tr}_{J}[\rho W_{\mu}(\bs{w})] \notag\\
    =& \frac{1}{2^t \te{Tr}(\rho \Pi_{\mu})} \sum_{T\in \mc{T}_{\mu}} 
    T\otimes \te{Tr}_{J}(\rho T)
\end{align}
where in the second and third lines, we have used the previous two results. The last line comes from the fact that $W_{\mu}(\bs{w})$ is just an element of $\mc{T}_{\mu}$ and the summation runs over all elements of $\mc{T}_{\mu}$.
\end{proof}

\subsection{The stabilizer state}
When the target state $\rho$ is a stabilizer state, the test operator can be further simplified. 
As indicated in the proof of Theorem \ref{theorem_test_operator}, when $\rho$ is a stabilizer state with stabilizer group $\mc{S}$,
$\Tr(\rho\Pi_{\mu})=|\mc{T}_{\mu}\cap \mc{S}|/|\mc{T}_{\mu}|$,
then the test operator $\Omega_{\mu}$ becomes,
\begin{align}
    \Omega_{\mu}
    =& \frac{1}{2^n |\mc{S}\cap \mc{T}_{\mu}|} \sum_{T\in \mc{T}_{\mu}} \sum_{S\in \mc{S}}  
    T\otimes \te{Tr}_{J}(TS). \label{test_operator_stabilizer_state}
\end{align}
Because $S, T$ are both Pauli strings, $\te{Tr}_{J}(S T)\ne 0$ if and only if $S$ and $T$ have the same Pauli matrices at the qubit sites $\{j_1,\cdots,j_t\}$, this motivates us to define the set 
\begin{align}
    \mc{R}_{\mu}:=\left\{S\in \mc{S}| \exists T\in \mc{T}_{\mu}, \te{Tr}_J(TS)\ne 0 \right\}.
    \label{overlap_group}
\end{align} 
The key point is that $\mc{R}_{\mu}$ is a stabilizer subgroup of $\mc{S}$, which has been known in \cite{Hayden_2016, Gross_2021}, then $\Omega_{\mu}$ is the projector into the stabilizer code of $\mc{R}_{\mu}$. This fact is summarized rigorously in the following corollary.

\begin{corollary}
For a stabilizer state $\rho$ with stabilizer group $\mc{S}$, the test operator can be rewritten as
\begin{align}
    \Omega_{\mu} 
    = \frac{1}{2^{n-t} |\mc{S}\cap \mc{T}_{\mu}|} \sum_{R\in \mc{R}_{\mu}} R
    =\frac{1}{|\mc{R}_{\mu}|} \sum_{R\in \mc{R}_{\mu}} R.
\end{align}
\end{corollary}
\begin{proof}
From the definition of $\mc{R}_{\mu}$, we can see that there is an one-to-one correspondence between the set of pairs $(T,S)\in \mc{T}_{\mu}\times \mc{S}$ satisfying $\te{Tr}_J(TS)\ne 0$ with the group $\mc{R}_{\mu}$, then 
    \begin{align}
    \Omega_{\mu}
    =& \frac{1}{2^n |\mc{S}\cap \mc{T}_{\mu}|} \sum_{T\in \mc{T}_{\mu}} \sum_{S\in \mc{S}}  
    T\otimes \te{Tr}_{J}(TS) \notag\\
    =& \frac{2^t}{2^n |\mc{S}\cap \mc{T}_{\mu}|} \sum_{R\in \mc{R}_{\mu}} R \notag\\
    =& \frac{1}{2^{n-t} |\mc{S}\cap \mc{T}_{\mu}|} \sum_{R\in \mc{R}_{\mu}} R,
\end{align}
so $\Omega_{\mu}$ is proportional to the projector into the stabilizer code of $\mc{R}_{\mu}$.
Since $\Omega_{\mu}$ is also a projector, known already in the last section, then we conclude that 
\begin{align}
    \Omega_{\mu}=\frac{1}{|\mc{R}_{\mu}|} \sum_{R\in \mc{R}_{\mu}} R. 
\end{align}
\end{proof}
From the above Corollary, we can easily derive the cardinality of the group, $|\mc{R}_{\mu}|=2^{n-t}|\mc{T}_{\mu}\cap \mc{S}|$. 
Compared to the similar results in \cite{Dangniam_2020}, our results actually encompass theirs as a special case of $t = n$.
For $\mc{R}_{\mu}$ is a subgroup of $\mc{S}$ and all elements of $\mc{S}$ commute with each other, therefore all test operators $\Omega_{\mu}$ commute with each other and can be diagonalizable simultaneously. 

Assuming the stabilizer group $\mc{S}$ has $n$ generators as $\{g_1,\cdots,g_n\}$, then we can construct a new stabilizer group $\mc{S}_{\bs{w}}$ by the generators $\{(-1)^{w_1}g_1,\cdots,(-1)^{w_n}g_n\}$ with each $\bs{w}\in \mb{Z}_2^n$, and the stabilizer state $\ket{\mc{S}_{\bs{w}}}$ is the common eigenstate of $g_i$ with eigenvalue $(-1)^{w_i}$ for $i=1,\cdots,n$. There are $2^n$ 
such stabilizer states $\ket{\mc{S}_{\bs{w}}}$ and they form an orthonormal basis in the Hilbert space, called the \textit{stabilizer basis} of $\ket{\mc{S}}$. 
Then we consider the matrix elements of $\Omega_{\mu}$ in the stabilizer basis,
\begin{align}
    \bra{\mc{S}_{\bs{w}'}}\Omega_{\mu} \ket{\mc{S}_{\bs{w}}} 
    =& \frac{1}{|\mc{R}_{\mu}|} \sum_{R\in \mc{R}_{\mu}} \bra{\mc{S}_{\bs{w}'}} R \ket{\mc{S}_{\bs{w}}} \notag\\
    =& \frac{1}{|\mc{R}_{\mu}|} \sum_{\bs{v}_R\in \mb{Z}_2^n} (-1)^{\bs{v}_R\cdot \bs{w}} \braket{\mc{S}_{\bs{w}'}}{\mc{S}_{\bs{w}}} \notag \\
    =& \frac{1}{|\mc{R}_{\mu}|} \sum_{\bs{v}_R\in \mb{Z}_2^n} (-1)^{\bs{v}_R\cdot \bs{w}} \delta_{\bs{w}',\bs{w}}
\end{align}
where in the second line we write $R=g_1^{v_{R,1}}\cdots g_n^{v_{R,n}}=\bs{g}^{\bs{v}_R}$ with $\bs{v}_R\in \mb{Z}_2^n$ and $\bs{g}^{\bs{v}_R} \ket{\mc{S}_{\bs{w}}} = (-1)^{\bs{v}_R\cdot \bs{w}} \ket{\mc{S}_{\bs{w}}}$.
Hence, all $\Omega_{\mu}$'s are diagonal in the stabilizer basis of $\mc{S}$ as
\begin{align}
    \Omega_{\mu}=\sum_{\bs{w}\in \mb{Z}_2^n} \gamma_{\mu,\bs{w}} \ket{\mc{S}_{\bs{w}}}\bra{\mc{S}_{\bs{w}}}
\end{align}
with $\gamma_{\mu,\bs{w}}$ denoting the diagonal element.
Further, the diagonal elements of $\Omega_{\mu}$ are either $0$ or $1$, 
\begin{align}
   \gamma_{\mu,\bs{w}}=\bra{\mc{S}_{\bs{w}}}\Omega_{\mu} \ket{\mc{S}_{\bs{w}}}
    = \left\{ \begin{array}{cc}
        1 & \te{if}~(-\mc{R}_{\mu})\cap \mc{S}_{\bs{w}}=\emptyset, \\
        0 & \te{otherwise},
    \end{array} \right.
\end{align}
which follows from Eq.\eref{overlap_projector} and the fact $\mc{R}_{\mu}$ is a subgroup of $\mc{S}$.

The special structure of $\Omega_{\mu}$ can simplify the optimization problem defined in Eq.\eref{optimization_problem} drastically. 
Having defined $(\mb{Z}_2^{2n})_t$ as the set of symplectic vector $\mu\in \mb{Z}_2^{2n}$ with weight $t$, the strategy operator $\Omega$ becomes
\begin{align}
    \Omega=& \sum_{\mu\in (\mb{Z}_2^{2n})_t} p_{\mu} \Omega_{\mu}  \notag\\
    =& \sum_{\boldsymbol{w}\in \mb{Z}_2^n} \sum_{\mu\in (\mb{Z}_2^{2n})_t} p_{\mu} \gamma_{\mu, \boldsymbol{w}}
    \ket{\mc{S}_{\boldsymbol{w}}}\bra{\mc{S}_{\boldsymbol{w}}} \notag\\
    =& \sum_{\boldsymbol{w}\in \mb{Z}_2^n} \gamma_{t, \boldsymbol{w}}
    \ket{\mc{S}_{\boldsymbol{w}}}\bra{\mc{S}_{\boldsymbol{w}}},
\end{align}
with 
\begin{align}
    \gamma_{t, \boldsymbol{w}} = \sum_{\mu\in (\mb{Z}_2^{2n})_t} p_{\mu} \gamma_{\mu, \boldsymbol{w}}
    = \sum_{\mu\in (\mb{Z}_2^{2n})_t} p_{\mu} \bra{\mc{S}_{\bs{w}}}\Omega_{\mu}\ket{\mc{S}_{\bs{w}}},
\end{align}
so the eigenvalues of $\Omega$ are $\gamma_{t, \boldsymbol{w}}$'s, which are just the summation of those probabilities $p_{\mu}$ such that the overlap of the corresponding test operator $\Omega_{\mu}$ with the stabilizer state $\ket{\mc{S}_{\bs{w}}}$ doesn't vanish. 
Since the target state $\ket{\mc{S}}\bra{\mc{S}}$ is the eigenstate of $\Omega$ with eigenvalue $1$, then $\gamma_{t, \boldsymbol{0}}=1$.
Hence, we only need to know which $\bs{w}$ let $\gamma_{t, \boldsymbol{w}}$ be the second-largest eigenvalue of $\Omega$. So the optimization problem Eq.\eref{optimization_problem} reduces to 
\begin{align}
   \min_{\{p_{\mu}\}}\max_{\bs{w}\in \mb{Z}_2^n, \bs{w}\ne \bs{0}} \sum_{\mu\in (\mb{Z}_2^{2n})_t} p_{\mu} \bra{\mc{S}_{\bs{w}}}\Omega_{\mu}\ket{\mc{S}_{\bs{w}}}. \label{optimization_problem_2}
\end{align}
The above optimization problem can be further reduced into a linear programming as done in \cite{Dangniam_2020}.

Although we have reduced the problem of finding the optimal strategy operator $\Omega$ to the optimization problem Eq.\eref{optimization_problem_2}, solving this problem is not an easy task as the number of $\mu$ with weight $t$ is exponentially large,
\begin{align}
    \binom{n}{t}3^t= \frac{3^t n!}{t!(n-t)!}.
\end{align}
Fortunately, the physical interesting target states have many symmetries, for example, 
the GHZ state. The symmetries of the target state can reduce the number of test operators significantly. Furthermore, taking into the consideration of the symmetries of the target state can also boost the spectral gap of the strategy operator \cite{Yu_2019,Liu_2019}.

Actually, there have been several explicit protocols for verifying the stabilizer state in \cite{Pallister_2018, Zhu_2019, Dangniam_2020}. 
So compared with finding the optimal strategy operator, sometimes we are more interested in the special strategy operator $\Omega_{\te{u}}$ corresponding to the uniform sampling, since we don't need to solve an optimization problem. 
For example, without considering the symmetries of the target state, we let all $p_{\mu}$'s are equal, i.e., 
\begin{align}
    p_{\mu} = \frac{1}{\binom{n}{t}3^t} = \frac{t!(n-t)!}{3^t n!},
\end{align}
then the eigenvalues of $\Omega_{\te{u}}$ are
\begin{align}
    \gamma_{t, \boldsymbol{w}}
    =\frac{t!(n-t)!}{3^t n!} \sum_{\mu\in (\mb{Z}_2^{2n})_t} \bra{\mc{S}_{\bs{w}}}\Omega_{\mu}\ket{\mc{S}_{\bs{w}}}. \label{naive_uniform_eiganvalue}
\end{align}
The second largest $\gamma_{t,\boldsymbol{w}}$ is the one such that the number of non-zero $\bra{\mc{S}_{\bs{w}}}\Omega_{\mu}\ket{\mc{S}_{\bs{w}}}$ is second-largest. Although for a general stabilizer state, the second-largest number is hard to compute, but for some special stabilizer state like GHZ state, it's possible to obtain such a number analytically.

\subsection{The GHZ state}
In this subsection, we consider one type of state with very special structure, the GHZ state.
The $n$-qubit GHZ state is defined as
\begin{align}
    \ket{\te{GHZ}_n}=\frac{1}{\sqrt{2}}(\ket{0}^{\otimes n}+\ket{1}^{\otimes n}),
\end{align}
it's stabilizer group $\te{Stab}(\ket{\mr{GHZ}_n})$ is generated by 
\begin{align}
    g_1=X^{\otimes n},\quad g_a=Z\otimes I^{\otimes(a-2)}\otimes Z \otimes I^{(n-a)}
\end{align}
for $2\le a \le n$, and the group elements of $\te{Stab}(\ket{\mr{GHZ}_n})$ can be classified into two types up to the permutation of $n$ qubits: $Z^{\otimes 2i}\otimes I^{\otimes (n-2i)}$, $(-1)^i Y^{\otimes 2i}\otimes X^{\otimes (n-2i)}$ for $i=0,1,\cdots,\lfloor n/2 \rfloor$.

It's obvious that $\ket{\te{GHZ}_n}$ is invariant under the action of any permutation $P\in \mc{P}_{n}$ of $n$ qubits, then if two different composite Pauli measurements are the same up to a permutation of $n$ qubits, they don't result in any difference in the physical sense, hence the probabilities $p_{\mu}$ for sampling such two measurements should be the same. 
However, for the labeling of $n$ qubits, the test operators for the composite Pauli measurements being the same up to some permutations actually are different, so we need to take the average of them together. 
More explicitly, two symplectic vectors $\mu, \mu'$ labeling two composite Pauli measurements are equivalent, denoted as $\mu\sim \mu'$, if they can be transformed into each other by applying some permutation $P\in \mc{P}_n$. Each equivalence class $[\mu]$ can be distinguished by the number of Pauli-$X$, $Y$ and $Z$ measurements, $(m_x, m_y, m_z)$. For each equivalence class $[\mu]$ labelled by $(m_x,m_y,m_z)$, we sum over all corresponding test operators and divide it by the number of test operators in this class as
\begin{align}
    \Omega_{[\mu]}=\frac{1}{|[\mu]|} \sum_{\mu\in [\mu]} \Omega_{\mu}
    =\frac{m_x!m_y!m_z!}{\binom{n}{t} t!}\sum_{\mu\in [\mu]} \Omega_{\mu},
\end{align}
with $m_x+m_y+m_z=t$. Then the strategy operator becomes
\begin{align}
    \Omega = \sum_{[\mu]}p_{[\mu]}\Omega_{[\mu]}
\end{align}
and the diagonal value is written as
\begin{align}
    \gamma_{t,\bs{w}}= \sum_{[\mu]}p_{[\mu]} \frac{1}{|[\mu]|} \sum_{\mu\in [\mu]}\gamma_{\mu,\bs{w}}.
\end{align} 
If we further consider the uniform sampling in the set of equivalence classes, then
$p_{[\mu]}=2/[(t+1)(t+2)]$, for the number of equivalence classes is $\binom{t+2}{t}$.
Hence, the diagonal value becomes
\begin{align}
    \gamma_{t,\bs{w}}=\frac{2}{(t+2)(t+1)} \sum_{[\mu]} \frac{m_x!m_y!m_z!}{\binom{n}{t} t!} \sum_{\mu\in [\mu]}\gamma_{\mu,\bs{w}}.
\end{align}

\subsubsection{$3$-qubit GHZ state}

Here we consider the $3$-qubit GHZ state, the generalization to the $n$-qubit case is straightforward.
When we apply the DPSO protocol \ref{alg:my_protocol} to the $3$-qubit GHZ state $\ket{\te{GHZ}_3}$, there are two different cases, corresponding to level-$1$ and level-$2$. 

\paragraph{Level-$1$}
According to the previous discussions, we need to compute all test operators $\Omega_{\mu}$ with $\mu$ having weight $2$.
For simplicity, we only compute the expressions of $\Omega_{K,\bs{l}}$ with $K=\{3\}$ ($J=\{1,2\}$) and $\bs{l}\in \{1,2,3\}^2$. The remaining test operators can be obtained through symmetry considerations via appropriate index permutations. 
Having obtained all test operators, we can consider two strategy operators corresponding to two different uniform sampling. 
\begin{itemize}
    \item The first one is the strategy operator corresponding to the naively uniform sampling, the expressions of whose eigenvalues are given in Eq.(\ref{naive_uniform_eiganvalue}). In this case, we conclude the spectral gap is $4/9$.
    \item The second one is the strategy operator corresponding to the uniform sampling over the equivalence classes. There are totally $6$ classes: $(m_x,m_y,m_z)=(2,0,0)$, $(0,2,0)$, $(0,0,2)$, $ (1,1,0)$, $(1,0,1)$ or $(0,1,1)$. Finally, we conclude the spectral gap is $1/2$.
\end{itemize}

\paragraph{Level-$2$}
In the level-$2$ protocol, we need to compute the test operators $\Omega_{\mu}$ with $\mu$ having weight $1$. As before, we consider two cases of uniform sampling. For the naive uniform sampling, we can conclude the spectral gap is $2/3$. While after considering the symmetries of $\ket{\mr{GHZ}_3}$, we obtain the spectral gap as $2/3$ too.

\subsubsection{$n$-qubit GHZ state}

The simple example of $\ket{\mr{GHZ}_3}$ leads us to a general result about the spectral gaps of the strategy operators of two uniform sampling cases. 
\begin{theorem}
     For two different uniform sampling over all symplectic vectors $\mu$ with weight $t=(n-r)$ and all symmetric equivalence classes $[\mu]$, the second-largest eigenvalue of the level-$r$ strategy operator $\Omega$ always corresponds to the eigenvector $\mc{S}_{(1,\bs{0}^{n-1})}$ in the stabilizer basis of $\ket{\mr{GHZ}_n}$. For the two cases, the second-largest eigenvalues are $1-(2/3)^{n-r}$ and $(n-r)/(n-r+2)$, then the spectral gaps are
     \begin{align}
         \left(\frac{2}{3}\right)^{n-r},\quad\frac{2}{n-r+2}
     \end{align}
     respectively. \label{theorem_GHZ}
\end{theorem}
The proof of this theorem is provided in the Appendix \ref{appendix_theorem_ghz}.
It's astonishing to find that after considering the symmetries of $\ket{\mr{GHZ}_n}$, the spectral gap of the strategy operator for the uniform sampling is boosted from exponentially small to the inverse linear for $n$. The weight of Pauli measurements used in the level-$r$ protocol is $t=n-r$, then the spectral gap becomes $2/(t+2)$, which is independent of the qubit number $n$! We also see that increasing the level $r$ doesn't improve the sample complexity. 

Compared with the known optimal protocols in the literature \cite{Pallister_2018, Li_2020, Dangniam_2020}, the sample complexity of our verification protocol is not optimal and grows linearly with the number of qubits $n$. However, this trade-off comes with the advantage that the protocol is easier to implement and can accommodate a broader class of quantum states simultaneously.
There are also some papers addressing the verification of GHZ states from the perspectives of self-testing and cryptography \cite{li2019selftestingsymmetricthreequbitstates, Hayashi2022}. These approaches are based on a fundamentally different assumption that the quantum devices or participants in the protocol may be untrusted, then the sample complexity of these protocols are much larger than our protocols, since we assumes that the measurement outcomes are fully reliable.




\section{Conclusion}

In this paper, we propose a new protocol, called DPSO protocol, for verifying of a pure state under the assumption of the i.i.d., which generalizes the ideas of the SOP introduced in \cite{10756060_huang} and \cite{li2025universalefficientquantumstate}. The DPSO protocol inherits the key advantage of the SOP, namely the ability to verify almost all Haar-random pure states, while extending its applicability to a broader class of states in a more natural manner. Moreover, our protocol achieves lower sample complexity compared to the SOP. In particular, we demonstrate that the DPSO protocol exhibits a concise structure when applied to stabilizer states, just as \cite{Dangniam_2020}.
Besides, we also formulate the SOP in the form of hypothesis testing and discuss the relationships between the PLM framework and the SOP in details. The SOP deals with the problem of QSV in a new different view. 
In general, the sample complexity of the DPSO (SOP) is worse than that of the PLM framework with respect to the infidelity $\epsilon$ and the spectral gap $\nu(\Omega)$. The DPSO protocol assumes the states produced by the quantum device is i.i.d., while the PLM framework can deal with more general scenarios where the states can be adversarial.
However, the advantage of the DPSO protocol is that it's possible to verify many types of states uniformly instead of designing different strategy operators in the PLM framework for different types of states.
 
In general, the verification of quantum states can be classified into two categories. The first involves verifying generic states in the Hilbert space, such as Haar random pure states. These generic states lack special structures, making verification challenging within the ordinary PLM framework, while the SOP or the DPSO protocol provides a general approach to address the problem. 
The second category pertains to verifying states with special structures. These special states often possess many symmetries, necessitating the design of explicit optimal or near-optimal protocols. For many types of special states, there have been some very efficient or optimal protocols in the original PLM framework.
Even within the SOP formalism or our protocol, it is crucial to select appropriate probability distributions for different Pauli measurements or to solve an optimization problem, in order to get an optimal protocol. 

Considering the current situation of the research of QSV, several challenges remain for the SOP or DPSO protocol to address in the future. First, it's possible that this protocol can handle more special states, beyond stabilizer states to include non-stabilizer states. Second, it's possible that the DPSO protocol with high level $r$ can also verify almost all Haar random pure states, as the level-$1$ SOP protocol \cite{10756060_huang}. Third, in the real-world experiments, the physical noise is unavoidable, then we need to consider the consequence of physical noise for the DPSO protocol. Fourth, it is also of interest to consider combining the SOP or DPSO protocols with self-testing techniques to address verification problems in scenarios similar to that discussed in \cite{Hayashi2022}.
Hence, Addressing these challenges is an important direction for future research. 
~\\


\textbf{Acknowledgements}: We thank Huangjun Zhu, Yunting Li, Zhihao Li, and Datong Chen for their helpful discussions, with special thanks to Yunting Li for providing the draft of her work with Huangjun Zhu.

\textbf{Author Contributions}: Xiaodi Li proposed the idea, performed the analytical and numerical computations, contributed to the interpretation of the results, and prepared the manuscript. 

\textbf{Conflicts of Interest}: The authors declare no conflicts of interest.
\bibliography{ref}

\appendix
\clearpage
\onecolumngrid

\section{Classical description of target states}
\label{appendix_classical_target}
In the DPSO protocol, the post-measurement state is computed according to Eq.\eref{post_measure_psi} of the manuscript. Two steps primarily determine the computational cost: the contraction of $|\boldsymbol{z}\rangle$ with $|\psi\rangle$, and the contraction of $|\widetilde{\phi}_{K,\boldsymbol{z}}\rangle$ with itself. Here, we take the matrix product state (MPS) as a representative example of tensor network states. Suppose the tensor network representation of $\ket{\psi}$ is
\begin{align}
\ket{\psi} = \sum_{b_1,\dots,b_n} A^{b_1}_{\alpha_1} A^{b_2}_{\alpha_1\alpha_2}\cdots A_{\alpha_{N-1}}^{b_N} \ket{b_1}\otimes \cdots \otimes \ket{b_N}.
\end{align}
Without loss of generality, we assume that Pauli measurements are applied to the first $n - r$ qubits. The computation of the contraction $\langle \boldsymbol{z}|\psi\rangle$ involves contracting the vector corresponding to $\ket{z_i}$ with the $i$th tensor $A^{b_i}$ for $1 \le i \le n - r$, followed by the contraction of the first $n - r$ tensors. Therefore, the computational complexity is $\mathcal{O}((n - r) D^2)$, where $D$ is the bond dimension.
The second step of the contraction of $|\widetilde{\phi}_{K,\boldsymbol{z}}\rangle$ with itself is just the ordinary inner product of two MPS, whose complexity is $\mathcal{O}(rD^3)$.
Hence, we conclude that the computational complexity of obtaining the post-measurement state is $\mathcal{O}((n - r) D^2 + r D^3)$, implying that the computation is efficient for MPS. 

Unfortunately, for other classes of states, computing the post-measurement state may not be efficient.
In contrast, for projected entangled pair states (PEPS), the exact contraction of two PEPS is known to be computationally hard. Therefore, the computation of the post-measurement state in this case is also exponentially hard.
For fermionic Gaussian states, the Jordan-Wigner transformation allows these states to be mapped to multi-qubit states. However, under this mapping, Pauli-$X$ and Pauli-$Y$ measurements become non-local operations and generally transform Gaussian states into non-Gaussian ones. Consequently, after applying such measurements, computing the post-measurement state of large-scale fermionic Gaussian systems may become computationally inefficient.

However, we can consider this problem from a different perspective.
Our assumption of the DPSO protocol that there exists a classical description of the target state is equivalent to the assumption of the query model in the SOP. However, there exists more stronger assumptions, for example, the query model in \cite{gupta2025singlequbitmeasurementssufficecertify}. 
Here we can generalize the query model of the SOP to allow access to the amplitudes of $\ket{\psi}$ in the following expansion:
\begin{align}
\ket{\psi} = \sum_{s_1,\cdots,s_n} c_{s_1,\cdots,s_n} \ket{s_1}\otimes \cdots\otimes \ket{s_n},
\end{align}
where each $s_i$ corresponds to an eigenstate of the Pauli-$X$, $Y$, or $Z$ operator.
Under this assumption, the first step—the contraction of $|\boldsymbol{z}\rangle$ with $|\psi\rangle$—requires $2^r$ queries to the model, resulting in a time complexity of $\mathcal{O}(2^r)$.
In the second step, computing the inner product of $|\widetilde{\phi}_{K,\boldsymbol{z}}\rangle$ with itself involves $2^r$ multiplications and summations of amplitudes, also yielding a time complexity of $\mathcal{O}(2^r)$.
Therefore, the overall time complexity for computing the post-measurement state in a multi-qubit system is $\mathcal{O}(2^r)$.

\section{Proof of Theorem \ref{theorem_GHZ}}
\label{appendix_theorem_ghz}

\begin{proof}
    We first consider the structure of the group $\Omega_{\mu}$ for different composite Pauli measurements with weight $t=n-r$.  Suppose the Pauli measurements $\sigma_{l_{j_1}}, \cdots,\sigma_{l_{j_{n-r}}}$ are applied to the set of qubits $J=\{j_1,\cdots,j_{n-r}\}$, the remaining qubits are denoted by $K=\{k_1,\cdots,k_r\}$.
    \begin{itemize}
        \item Suppose there are $h$ Pauli-$Z$'s are applied to the set of qubits $Q=\{q_1,\cdots,q_h\}\subset J$ with $1\le h\le t$. The stabilizer group corresponds to the measurement is
        \begin{align}
            \mc{T}_{\mu}=\{\sigma_{l_{j_1}}^{x_1} \otimes \cdots \otimes \sigma_{l_{j_{n-r}}}^{x_h}: \bs{x}\in \mb{Z}_2^{n-r} \}.
        \end{align}
        According to the structures of the element of $\mc{S}=\te{Stab}(\ket{\te{GHZ}_n})$, only the $T\in \mc{T}_{\mu}$ consisting solely of the identity $I$ or Pauli-$Z$'s contribute to the group $\mc{R}_{\mu}$. Then we can specify the group $\mc{R}_{\mu}$ as
        \begin{align}
            \mc{R}_{\mu}=\{Z_{q_1}^{x_1}\otimes \cdots \otimes Z_{q_h}^{x_h}\otimes Z_{k_1}^{x_{h+1}}\otimes \cdots \otimes Z_{k_r}^{x_{h+r}}\otimes I^{\otimes (n-h-r)}:\bs{x}\in\mb{Z}_2^{h+r}~\te{and}~ \te{w}(\bs{x}) ~\te{is even} \},
        \end{align}
        where $\te{w}(\bs{x})$ is the number of $1$ in $\bs{x}$. The number of such $\mc{R}_{\mu}$'s is
        \begin{align}
             \binom{n}{n-r}\left[ \sum_{h=1}^{n-r} \binom{n-r}{h} 2^{n-r-h} \right] =\binom{n}{n-r}(3^{n-r}-2^{n-r}).
        \end{align}
        
        According to the definition of $\mc{S}_{\bs{w}}$, the element of the group $\mc{S}_{\bs{w}}$ with $\bs{w}=(1,\bs{0}^{n-1})$ has the structure $(-1)^{w_1}\bs{g}^{\bs{v}}$ with $\bs{v}\in \mb{Z}_2^{n}$. Considering the definitions of the generators of $\te{Stab}(\ket{\te{GHZ}_n})$, the elements of $\mc{S}_{(1,\bs{0}^{n-1})}$ consisting of solely $I$ or $Z$ don't have extra minus sign, so all test operators corresponding to the above group $\mc{R}_{\mu}$ will have non-vanishing overlap with $\ket{\mc{S}_{\bs{w}}}\bra{\mc{S}_{\bs{w}}}$.
        
        However, as for $\mc{S}_{\bs{w}}$ with $\bs{w}\ne (1,\bs{0}^{n-1})$, there exists at least one generator $g_i$ with $i>1$ having extra minus sign, then
        the elements of $\mc{S}_{\bs{w}}$ consisting of solely $I$ or $Z$ may have extra minus sign, then not all $\mc{R}_{\mu}$'s contribute. We define $I(\bs{w})=\{1\le i\le n: w_i=1\}$ with $|I(\bs{w})|\ge 1$. (1) If $I(\bs{w})\subset J$, then the number of contributing $\mc{R}_{\mu}$ is 
        \begin{align}
            \sum_{h=1}^{t-|I(\bs{w})|} \binom{t-|I(\bs{w})|}{h} 2^{t-h} \le \sum_{h=1}^{t-1} \binom{t-1}{h} 2^{t-h}
            =2(3^{t-1}-2^{t-1}),
        \end{align}
        since the set $Q\cup K$ of $\mc{R}_{\mu}$ can't contain $I(\bs{w})$.
        (2) If $I(\bs{w})\subset K$, then the number of contributing $\mc{R}_{\mu}$ is 
        \begin{align}
            \sum_{l=1}^{\lfloor t/2 \rfloor} \binom{t}{2l} 2^{t-2l}= \frac{1}{2}(3^t+1)-2^t,
        \end{align}
        where we have used the fact
        \begin{align}
            \frac{1}{2}[(2-1)^t+(2+1)^t]=\frac{1}{2} \left[\sum_{k=0}^t \binom{t}{k}(-1)^k 2^{t-k} + \sum_{k=0}^t \binom{t}{k} 2^{t-k}  \right]
            = \sum_{l=0}^{\lfloor t/2 \rfloor} \binom{t}{2l} 2^{t-2l}.
        \end{align}
        So the largest number of test operators have non-vanishing overlap with $\ket{\mc{S}_{\bs{w}}}\bra{\mc{S}_{\bs{w}}}$ is 
        \begin{align}
           2\binom{n-1}{t-1}(3^{t-1}-2^{t-1})+\binom{n-1}{t}\left[\frac{1}{2}(3^t+1)-2^t \right].
        \end{align}

        \item Suppose there is no Pauli-$Z$ in the composite Pauli measurement but $h$ Pauli-$Y$, which are applied to the set of qubits $Q=\{q_1,\cdots,q_h\}\subset J$ with $1\le h\le t$.
        According to the structures of the element of $\te{Stab}(\ket{\te{GHZ}_n})$, only the $T\in \mc{T}_{\mu}$ consisting solely of the identity $X$ and Pauli-$Y$'s contribute to the group $\mc{R}_{\mu}$. 
        The total number of such $\mc{R}_{\mu}$ is 
        \begin{align}
            \binom{n}{t} 2^t.
        \end{align}
        Such $\mc{R}_{\mu}$ can't contribute to the term $\te{Tr}(\Omega_{\mu}\ket{\mc{S}_{\bs{w}}}\bra{\mc{S}_{\bs{w}}})$ with $\bs{w}= (1,\bs{0}^{n-1})$, but it's possible for other $\bs{w}$.
        
    \end{itemize}

    Combing the above discussions, the total number of $\mu$, which will contribute to $\te{Tr}(\Omega_{\mu}\ket{\mc{S}_{\bs{w}}}\bra{\mc{S}_{\bs{w}}})$ with $\bs{w}\ne (1,\bs{0}^{n-1})$, is upper bounded by
    \begin{align}
        2\binom{n-1}{t-1}(3^{t-1}-2^{t-1})+\binom{n-1}{t}\left[\frac{1}{2}(3^t+1)-2^t \right]+ \binom{n}{t} 2^t
        =\binom{n}{t}\left[ \frac{3n+t}{2n}3^{t-1}+\frac{n-t}{2n} \right].
    \end{align}
    We can see that the above term is smaller than $\binom{n}{t}(3^{t}-2^{t})$, which is the number of $\mu$'s contributing to $\te{Tr}(\Omega_{\mu}\ket{\mc{S}_{\bs{w}}}\bra{\mc{S}_{\bs{w}}})$ with $\bs{w}= (1,\bs{0}^{n-1})$. 
    So we conclude that the state $\ket{\mc{S}_{\bs{w}}}$ corresponds is the eigenvector of the uniform strategy operator $\Omega_U$ with the second-largest eigenvalue, and the eigenvalue is
    \begin{align}
        \frac{\binom{n}{t}(3^{t}-2^{t})}{\binom{n}{t}3^{t}}
        =1-\left(\frac{2}{3}\right)^t.
    \end{align}.

    As for the of uniform sampling over the equivalence classes, since all weight-$t$ Pauli measurement containing at least one Pauli-$Z$ contribute to the $\te{Tr}(\Omega_{\mu}\ket{\mc{S}_{\bs{w}}}\bra{\mc{S}_{\bs{w}}})$ with $\bs{w}=(1,\bs{0}^{n-1})$, the operation of averaging in each equivalence class doesn't change the fact each class also contribute $1$ to $\te{Tr}(\Omega_{\mu}\ket{\mc{S}_{\bs{w}}}\bra{\mc{S}_{\bs{w}}})$ with $\bs{w}=(1,\bs{0}^{n-1})$. While for other terms $\te{Tr}(\Omega_{\mu}\ket{\mc{S}_{\bs{w}}}\bra{\mc{S}_{\bs{w}}})$ with $\bs{w}\ne (1,\bs{0}^{n-1})$, only part equivalence classes will contribute. So the state $\ket{\mc{S}_{\bs{w}}}$ corresponds is also the eigenvector of the uniform strategy operator. And the number of equivalence classes of weight $t$ Pauli measurement containing at least one Pauli-$Z$ is
    \begin{align}
        \binom{t-1+3-1}{t-1},
    \end{align}
    then the second-largest eigenvalue is
    \begin{align}
        \binom{t+1}{t-1}/\binom{t+2}{t}=\frac{t}{t+2}.
    \end{align}
     
\end{proof}

\end{document}